%% file: v5.tex
\documentclass{elsarticle}

\usepackage{lineno,hyperref}

\journal{Transportation Research Part B: Methodological}

\usepackage{indentfirst}

\usepackage{tikz}
\usetikzlibrary{arrows, automata, decorations.pathmorphing, backgrounds, positioning, fit, petri, matrix, patterns}

\usepackage[linesnumbered,vlined,ruled]{algorithm2e}

\usepackage{amsmath,amsthm}
\usepackage{amssymb}
\usepackage{graphicx,setspace}
\usepackage{color}
\usepackage[english]{babel}
\usepackage{mathrsfs,dsfont}
\usepackage{array}
\usepackage{adjustbox}
\usepackage{multirow}
\usepackage{hyperref}
\usepackage{subcaption}

\usepackage{todonotes}
\usepackage{fullpage}

\newtheorem{theorem}{\bf Theorem}
\newtheorem{lemma}{\bf Lemma}

\newtheorem{proposition}{\bf Proposition}

\newtheorem{assumption}{Assumption}

\newtheorem{remark}{Remark}
\newtheorem{definition}{Definition}

\def\QED{~\rule[-1pt]{5pt}{5pt}\par\medskip}
\renewenvironment{proof}{{\bf Proof: \ }}{ \hfill \QED}

\newcommand{\transpose}{\mathsf{T}}

\newcommand{\ind}[1]{\mathds{1}_{\{#1\}}}
\newcommand{\beq}[1]{\begin{linenomath} \begin{equation} #1 \end{equation}\end{linenomath}}
\newcommand{\beqq}[1]{\begin{linenomath} \begin{equation*} #1 \end{equation*}\end{linenomath}}
\renewcommand{\Re}{\mathbb{R}}

\renewcommand{\forall}{\text{ for all }}

\begin{document}
	
	\begin{frontmatter}
		
		\title{Incentive Design and Profit Sharing in Multi-modal Transportation Networks}

		\author[mymainaddress]{Yuntian Deng\corref{mycorrespondingauthor}}\ead{deng.556@osu.edu}
		
		\author[mymainaddress]{Shiping Shao}\ead{shao.367@osu.edu}
		
		\author[mythirdaddress]{Archak Mittal}\ead{amittal9@ford.com}
		
		\author[mythirdaddress]{Richard Twumasi-Boakye}\ead{rtwumasi@ford.com}
		
		\author[mythirdaddress]{James Fishelson}\ead{jfishels@ford.com}
		
		\author[mymainaddress]{Abhishek Gupta}\cortext[mycorrespondingauthor]{Corresponding author}\ead{gupta.706@osu.edu}
		
		\author[mymainaddress,mysecondaryaddress]{Ness B. Shroff}\ead{shroff.11@osu.edu}

		\address[mymainaddress]{Department of Electrical and Computer Engineering, The Ohio State University, Columbus, OH 43210, USA}
		\address[mythirdaddress]{Ford Motor Company, 20000 Rotunda Dr, Dearborn, MI 48124, USA}
		
		\address[mysecondaryaddress]{Department of Computer Science and Engineering, The Ohio State University, Columbus, OH 43210, USA}

		\begin{abstract}
			We consider the situation where multiple transportation service providers cooperate to offer an integrated multi-modal platform to enhance the convenience to the passengers through ease in multi-modal journey planning, payment, and first and last mile connectivity. This market structure allows the multi-modal platform to coordinate profits across modes and also provide incentives to the passengers. Accordingly, in this paper, we use cooperative game theory coupled with the hyperpath-based stochastic user equilibrium framework to study such a market. We assume that the platform sets incentives ({ price discount or excess charge on passengers}) along every edge in the transportation network. We derive the continuity and monotonicity properties of the equilibrium flow with respect to the incentives along every edge. The optimal incentives that maximize the profit of the platform are obtained through a two time-scale stochastic approximation algorithm. We use the asymmetric Nash bargaining solution to design a fair profit sharing scheme among the service providers. We show that the profit for each service provider increases after cooperation on such a platform. {Finally,} we complement the theoretical results through two numerical simulations.
		\end{abstract}
		
		
	\end{frontmatter}
	

	\section{Introduction}
	Myriad options are now available for passengers to commute between different places. At the moment, all these services are owned and operated by distinct competitors, where every firm is trying to attract as many customers as possible. Such competition, while beneficial to passengers economically, is resulting in a huge cost to society and the city. For example, there is an increase in traffic congestion due to increased vehicle miles traveled by ridesharing services \cite{li2016demand}, interactions between various modes of transportation, and a shift from sustainable public transportation to individual ridesharing vehicles. 
	
	Meanwhile, one of the grand challenges outlined in the European Union's study on the transportation sector \cite{website2019} is to develop incentive schemes to encourage cooperation between various modes of transportation. {Such cooperative multimodal transportation,  not only reduces the traffic congestion \cite{litman2017introduction},  but also is beneficial for the environment and public health \cite{zhang2018exploring}}. There is also significant interest in integrating shared mobility with the public transit systems throughout the country, as indicated by the U.S. Department of Transportation \cite{mccoy2018integrating}. Public-Private Partnerships (PPPs) attract a lot of interest in various transportation projects \cite{button2016public}. 
	
	A natural question that arises is how service providers decide whether or not to cooperate. Is it possible that their profits decrease after cooperation? More interestingly, how would the total profit be shared if they agree to cooperate? More generally, in a complex network with multiple modes of transportation, how can multiple service providers act in a manner that is beneficial to the city, in terms of lower congestion, lower emissions, efficient use of road capacity, etc. Typically, if the price and convenience of service go down, then its consumption increases, which is termed as induced demand. Thus, it becomes unclear if the increase in consumption is beneficial to the larger public or not. For instance, if the cost of ride-hailing services decreases, then there will be increased usage of vehicles carrying one or two passengers, which will in turn increase the congestion in the city. 
	
	Accordingly, the goal of this paper is to present a theoretical framework to analyze the market aspect of possible cooperation between various competing modes of transportation. We explicitly account for induced demand attributed to the cooperation. Our optimization formulation is such that the total profit of the platform increases after cooperation. We then use asymmetric Nash bargaining solution to design the profit-sharing scheme among the individual service providers. To the best of our knowledge, this is the first study to devise an algorithm that computes the passengers' incentives in a multi-modal transportation system and designs the associated profit sharing scheme for the service providers.

	\subsection{Prior Work}
	With the popularity of shared mobility such as bike, scooter, car, etc., multi-modal transportation is becoming more common \cite{litman2017introduction}, where at least two different modes of transport are involved. In this work, we study how the equilibrium flow changes in a multi-modal transportation network when service providers cooperate and incentives are applied, as well as the resulting profit-sharing scheme.   
	As such, the problem studied here draws inspiration from various distinct strands of recent research on transportation markets--ones that propose induced flow models, pricing of services, and collaborative transportation in which service providers collaborate--which we review below.
	
	\subsubsection{Hyperpath}
	The concept of a hypergraph and optimal hyperpaths have been extensively researched in the transportation literature. Hyperpaths are efficient ways for considering multi-modal networks and trips as unlike the perpetual viability of a monomodal path, multi-modal paths must respect the constraints of available modal options to achieve a viable path \cite{lozano2002shortest}. The constrained hyperpaths which present the constituent paths within the hyperpaths are viable, thus can be taken by the user. Previous literature considers hyperpaths from the perspective of a single mode of public transportation which is analogous to a monomodal hyperpath \cite{nguyen1989hyperpaths}, while others expand the problem to find a solution for multi-modal networks with mode transfers at nodes \cite{ziliaskopoulos2000intermodal}. Studies also address solving traffic assignments particularly for transit models \cite{marcotte1998hyperpath}, and also with a focus on the least cost for transit networks \cite{verbas2015finding}.  While existing tools and efforts reveal the depth of knowledge in the area, advanced methodologies for integrating { incentive designing and }cooperative profit sharing in the hyperpaths framework is understudied. This aspect requires more dynamism in determining the optimal hyperpath for a user, based on multiple service provider offerings and incentivizing the competing providers to cooperate.
	
	\subsubsection{Induced demand}
	{ Induced demand, is defined as a movement along the travel demand curve \cite{lee1999induced}. } For example, additional capacity stimulates corresponding increases in demand, which has been supported by recent evidence \cite{website2018}. Travel demand is the result of a combination of exogenous factors and endogenous factors, where the former determines the location of the demand curve, and the latter determines the price-volume point along the demand curve. The short-run elasticity considers the situation where some long-time characteristics do not change, such as the size and fuel efficiency of vehicles, locations of employment, and population growth. The long-run, on the other hand, considers both exogenous and endogenous conditions and usually is longer than one year. { Induced demand is also named as generated traffic \cite{litman2001generated}, which describes the additional vehicle travel that results from a transportation improvement.} { Two sources of induced demand is distinguished in \cite{hymel2010induced}:} those that occur in undeveloped areas when new locations are made more accessible, and those that occur in urban areas because increased capacity tends to reduce congestion, which attracts more traffic. { In this paper, we focus on the latter, where the  integrated multi-model platform attracts more traffic, as it enhances the convenience to passengers and reduces the route costs. }

	\subsubsection{Pricing of transportation platforms}
	As Uber and Lyft are offering dynamic pricing on the ridesharing services, much attention has been paid to the pricing policies of ride-hailing platforms \cite{agatz2012optimization}, which motivates us to design a link incentive mechanism for multi-modal transportation networks. \cite{banerjee2016pricing} develops a general approximation framework for designing pricing policies in shared vehicle systems, which provides the first efficient algorithm with rigorous approximation guarantees for throughput and revenue. In \cite{ma2018spatio}, a complete-information, discrete-time, multi-period, multi-location model is proposed and they also show that following the mechanism's dispatches forms a subgame-perfect equilibrium among the drivers under the Spatio-temporal pricing mechanism. Further, a multi-region and multi-modal transportation system with park-and-ride facilities are proposed in \cite{liu2017doubly}. Given that the travelers can adjust their mode choices from day to day, an adaptive mechanism is developed to update parking pricing (or congestion pricing) from period to period. Under a pricing model with detour-based discounts for passengers, \cite{biswas2017impact} provides an efficient solution to maximize profit for commercial ridesharing platforms.   { However, these models only consider the interest of specific companies (e.g., Uber/Lyft), which may drive the entire system into an inefficient state and benefit these companies at the expense of social efficiency. In this paper, we consider the interest of whole multi-model system, including  taxi, bus, subway and bike etc, and develop policies beneficial to the whole society. }
	
	\subsubsection{Collaborative transportation}
	Collaboration between two or more agents is becoming an important approach to finding efficient solutions in improving transportation service quality in big cities. { A comprehensive survey on cost allocation methods  is provided in  \cite{guajardo2016review}}. In \cite{agarwal2010network}, authors study alliance formation among carriers in liner shipping and design an allocation mechanism, which provides side payments to the carriers to motivate them to act in the best interest of the alliance. In \cite{dai2012profit}, a carrier collaboration problem in the pickup and delivery service is studied and three profit allocation mechanisms are provided, which are based on the Shapley value, the proportional allocation, and the contribution of each carrier. In \cite{yang2016collaborative},   the Nash bargaining model is adopted to determine profit allocation and a transfer payment contract including fixed and variable fees are characterized between two logistics service providers for collaborative distribution. {Our model generalizes this result to multiple service providers through an asymmetric Nash bargaining solution.}
	
	{ 
		\subsubsection{Pricing with Stochastic User Equilibrium}
		
		Link-based incentives is widely discussed with Stochastic User Equilibrium  (SUE). For a given demand, the existence of link tolls is shown in \cite{smith1995existence}, which drives the SUE link flow to a system optimum, such as minimizing average user costs.  Based on this,   marginal-cost link  toll is demonstrated as a good solution under logit-based SUE towards system optimum \cite{yang1999system}.  However, it is rarely implemented in practice, because it requires every link to be charged. Instead, four alternative toll solutions are proposed through different objectives under fixed or elastic demand \cite{liu2014toll}. These alternatives obtain the same link flows with marginal cost toll, and significantly reduce the number of charged links and total charged toll.  However, the marginal-cost toll cannot be negative, as it is set with the derivative of increasing cost function with respect to link flow, e.g., equation (18) in \cite{liu2014toll} according to marginal cost functions. Besides, in the elastic demand case (SUE), there is actually no active constraint because of decreasing demand function (Equation (55) in \cite{liu2014toll}). Our work generalizes current results through addressing both challenges: i) we allow negative tolls (price discount for passengers) without using marginal cost function  and ii) we introduce general SUE constraints for elastic demand. Further, our model works for  general nonlinear objective functions and we provide an efficient algorithm to solve the problem.}
	
	{ 
		\subsubsection{Multimodal mobility systems and MaaS}
		Multimodal mobility systems are an emerging service that leverages both traditional transit systems and emerging Mobility as a Service (MaaS) applicants.  Focusing on maximizing the overall welfare, several control patterns are introduced in \cite{luo2021multimodal}, including transit service, rebalancing flows and passenger's mode choice. The objective function is linear with respect to these control parameters and decomposed into three problems with approximation, under a fixed demand. Instead of pricing on travel modes, an auction-based online MaaS mechanisms is introduced in 	\cite{xi2020incentive}, where users can submit several bids with different travel time and willings to pay. This problem is then modeled as mixed-integer linear programming to maximize social welfare by allocating mobility resources to users.  As the road space is limited, a design of allocation of ubrban road capacity in the multimodal system is discussed in 	
		\cite{zhang2018systematic}, where private and public transport modes have separated road space. A steady-state equilibrium traffic
		characteristics is derived. Further, there are other approaches that support MaaS, such as trip planning, fare integration and data sharing  \cite{shaheen2020mobility}. Different from previous literature,  this work focuses on maximizing the overall profit of the multi-modal platform through link incentives, under some passenger-friendly constraints. }

	\subsection{Key Contributions of This Paper}
	In this paper, we model the economic interactions over a platform comprising of a coalition of distinct service providers and passengers in a multi-modal transportation network. Incentives are designed to encourage competing providers to cooperate and attract passengers to take multi-modal routes. The main contributions of this paper are summarized below. 
	\begin{itemize}
		\item We first recall the main result of \cite{cantarella1997general} that establishes the stochastic user equilibrium in the multi-class hyperpath-based multi-modal transportation network. That paper establishes the existence and uniqueness of the equilibrium flow under elastic demand. We also show that the equilibrium flow is continuous with respect to the link incentives and monotonic with respect to link cost under certain assumptions on the link cost function and travel demand function.
		
		\item We pose the platform's problem of inducing cooperation among the service providers as a profit maximization problem under {passenger-friendly} constraints. We then develop a two time-scale stochastic approximation algorithm to determine the optimal incentives and associated equilibrium flow. The convergence of the {proposed} algorithm is proved under certain assumptions in Theorem \ref{thm:twotimescale}. A large scale numerical simulation with multiple origin-destination pairs shows the effectiveness of our algorithm, as both flow error and incentive error go to zero as the number of iteration grows.
		
		\item We leverage an asymmetric Nash bargaining solution to design a fair profit-sharing scheme for service providers after cooperation in Theorem \ref{thm:sharing}. We {provide} a closed-form solution to derive the profit shared with each service provider and show that each provider improves its total profit through cooperation.
	\end{itemize}

	The rest of the paper is organized as follows. The system model is presented in Section II. Section III discusses the stochastic user equilibrium under different incentives. The properties of the equilibrium are stated in Section IV. Section V discusses the optimal pricing and the two time-scale algorithm. The profit-sharing scheme is stated in Section VI. Section VII shows the numerical simulation. Finally, Section VIII concludes the paper.

	\section{System Model}
	In this section, we state the system model for a multi-modal transportation network with several modes of transportation and service providers. Within one origin-destination pair,  passengers have multiple choices, considering financial costs, travel time, and level of inconvenience. If all the service providers cooperate and are incorporated by one platform, then the systematical quality-of-service (congestion, passenger's satisfaction, profit) can be largely improved by intelligently setting the incentives.

	The transportation network is modeled as a graph  $\mathcal{G}:=(\mathcal{Z},\mathcal{L})$, where $\mathcal{Z}$ denotes the set of nodes and $\mathcal{L}=\{l_{ij}|i,j \in \mathcal{Z} \}$ denotes the set of directed links connecting two nodes in $\mathcal{Z}$. Each link $l_{ij}$ represents that the passenger boards at node $i$ and alights at node $j$. We are interested in all the possible paths that start at the origin (node $o$) and ends at the destination (node $d$) with any combination of modes\footnote{In a real-world implementation, the paths that are too cumbersome to take can be removed from the consideration.}. Throughout the paper, we assume multi-origin and multi-destination traffic flows. Figure \ref{fig:hyperpath} depicts a simple network with a single origin and destination.
	
	\begin{figure}[bth]
		\centering
		\scalebox{0.6}{\input{hyperpath_fig}}
		\caption{\label{fig:hyperpath} Multi-modal transportation network with price and travel time on each link.}
	\end{figure}
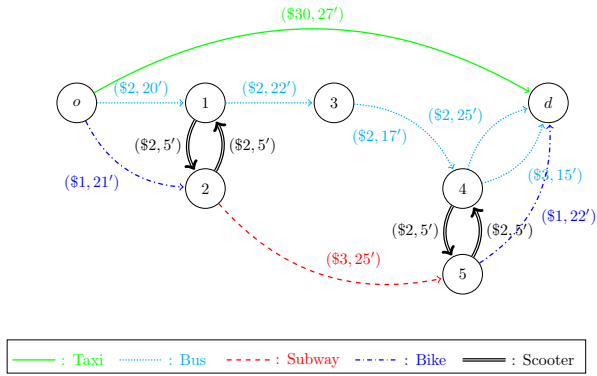

	\subsection{Model for Multi-modal Transportation Network}
	In this subsection, we introduce the model of a multi-modal transportation network and state the definition of link, path, and hyperpath, respectively.
	
	\subsubsection{Link}
	As shown in Figure \ref{fig:hyperpath}, each link represents one specific mode of transportation from one node to another in the multi-modal transportation network.
	
	Let $L= |\mathcal{L}|$ denote the total number of links in the system. Rewrite all the links as $\mathcal{L} =\{l_1, l_2, \ldots, l_L\}$. We assume that flow (number of passengers passing by) on link $l_i \in \mathcal{L}$ is $f_i, i \in \{1, \ldots, L\}$.  We set system link flow vector as $\vec f = [f_1, \ldots, f_L]^\transpose \in \Re ^L$. The associated link cost $c_i(\vec f)$ is defined as a function of flow, which is congestion-dependent and can be decoupled as financial cost (price) $c_{p,i}(\vec f)$ and travel time cost $c_{t,i}(\vec f)$. The Jacobian of $c_i(\vec f)$ can be asymmetric, which depends on flows of other links. We set  link 
	cost function as $C: \Re^L \rightarrow \Re^L$ and  $C(\vec f) = [c_1(\vec f), \ldots, c_L(\vec f)]^\transpose$. 
	
	\subsubsection{Path}
	Enumerate all the possible paths from node $o$ to node $d$ exhaustively. We define the set of all the possible paths as $\mathcal{N} \subseteq \mathcal{G}$. The elementary path is defined as follows:
	\begin{definition}
		$n:=(\mathcal Z_n,\mathcal L_n)\in \mathcal{N}$, where $\mathcal Z_n \subseteq \mathcal Z$, $\mathcal L_n \subseteq \mathcal{L}$, is an elementary path (usually called path)  if
	\end{definition}
	\begin{enumerate}
		\item $n$ starts from node $o$ and ends at node $d$;
		\item $n$ is acyclic and $n \neq \emptyset$.
	\end{enumerate}
	Passengers can follow any path $n \in \mathcal{N}$ in order to arrive at node $d$ starting from node $o$ and usually we require that there is at least one path, i.e. $|\mathcal{N}|\geq 1$. For example, in Figure \ref{fig:hyperpath}, passengers can ride a shared bike to node $2$, then take a subway to node $5$ and ride another shared bike towards the destination node $d$. 
	
	Denote the set of elementary paths as $\mathcal{N}$ and $N=|\mathcal{N}|$ is the total number of elementary paths. Let matrix $A \in \{0,1\}^{L*N}$ denotes the mapping from links to paths such that $a_{ln} = 1$ only when link $l \in \mathcal{L}$ belongs to path $n \in \mathcal{N}$.

	\subsubsection{Hyperpath (route)}
	The choice of links actually followed from origin to destination comes from both the pre-trip choice (choice made before trip starts) and the en-route choice (choice made during the trip). We model the pre-trip choice with the Random Utility model in the next subsection and the en-route choice in this subsection with a hyperpath framework, which succinctly describes the multiple paths from source to destination along with the probability with which these paths will be taken. 
	
	\begin{definition}
		Consider a directed subgraph and an associated probability vector, denoted by $m := (\mathcal Z_m, \mathcal L_m, \pi_m)$, where $\mathcal Z_m \subseteq \mathcal Z$, $\mathcal L_m \subseteq \mathcal L$, and $\pi_m=(\pi_{i,l,m})\in[0,1]^{|\mathcal L_m|}$. Then $m$ is a hyperpath (route) if:
		\begin{enumerate}
			\item Each path in $m$ is acyclic and $m \neq \emptyset$;
			\item For every node $i\in \mathcal Z_m -\{o,d\}$, there exists an elementary path $n$ traversing through $i$;
			\item The link diversion probability vector satisfies \beqq{\sum_{l \in \mathcal L_m } \pi_{i,l,m}=1,  \forall i \in \mathcal Z_m -\{d\}, i \textit{ is the origin of link } l} \beqq{\pi_{i,l,m} \geq 0,  \forall l\in \mathcal L_m}
		\end{enumerate}
	\end{definition}

	In other words, a hyperpath is a collection of multiple elementary paths along with an associated probability distribution.
	For example, as shown in Figure \ref{fig:hyperpath}, consider all the bus links as one hyperpath. Passengers can take the bus at node $o$, 1, 3, 4 and $d$. Node $4$ is a diversion node with $\pi_{4,l_1}+\pi_{4,l_2}=1$. Passengers at node 4 in this hyperpath will choose the upper link $l_1$ with probability $\pi_{4,l_1}$ and the lower link $l_2$ with probability $\pi_{4,l_2}$. Usually $\pi$ is pre-determined and derived from a schedule or frequency-based system such as bus and subway.
	
	Let $e_{nm}$ denote the en-route choice probability of path $n$ belonging to hyperpath $m$.  We have $e_{nm}=\prod_{l \in n} \pi_{i,l,m}$, i.e., it is the product of all link diversion probabilities along path $n$ within hyperpath $m$. In other words, with probability $e_{nm}$, passengers in hyperpath $m$ will follow the elementary path $n$. Let matrix $E=[e_{nm}]\in \{0,1\}^{N\times M}$ denote the relationship between elementary path and hyperpath.

	Let us index the hyperpaths(routes) $\mathcal{M}$ with $1,\ldots, M$, where $M=|\mathcal{M}|$ is the total number of routes. Consider the vector $\vec x=[x_1,...,x_M]^\transpose \in \Re^M$, which denotes the flow on each route. 
	
	Let matrix $B \in [0,1]^{L\times M}$ denotes the mapping from links to routes (traversing probability), such that route $m$ traverses link $l$ with traversing probability $b_{lm}$ where $0 \leq b_{lm} \leq 1$. As $b_{lm}=\sum_{n} a_{ln} e_{nm}$, we have $B=AE$. Therefore, the link flow vector follows 
	\beq{\vec f = B\vec x}
	
	For each route $m \in \mathcal{M}$, let $v_m(\vec x)$ denote the utility flow function on route $m$. This utility includes the real cost such as the price of the total trip, the time cost in equivalent currency, and notional costs such as the mode transfer cost, potential accident costs, and the cost associated with inconvenience due to the presence of traffic on that link (number of passengers in the bus/subway, traffic jams, etc.).  Let the route specific cost (transfer cost or cost of inconvenience) be $\vec c_s= [c_{s,1}, \ldots, c_{s,M}]^\transpose \in \Re^M$, which  depends on the topology. For any link cost vector $\vec c \in \Re^L$, define the route utility function as $V:\Re^L \rightarrow \Re^M$ and $V(\vec c)=[v_1(\vec x),\ldots,v_M(\vec x)]^\transpose$.  We have the route utility as follows
	\beq{\label{equ:utility}
		V(\vec c)= -\beta B^\transpose \vec c - \vec c_s }
	where $\beta$ is a constant. The first part of utility is the weighted sum of flow-dependent link costs along the route, and the second part of utility is the route specific cost, which only depends on the topology of the transportation network.

	\subsection{Random Utility Choice Model}
	In this subsection, we use the Random Utility Choice theory to model the pre-trip choice of passengers in this system.
	
	We assume that passengers make decisions about which route (hyperpath) they will take at the start of the trip and they do not change route during the trip. The en-route choice defined in the last subsection does not affect the chosen route but only affects which link to take at a node that belongs to the route. For example,  as shown in Figure \ref{fig:hyperpath}, a passenger can make a choice at the start of the trip. Once he chooses to take a route $o \rightarrow{} 1 \rightarrow{} 3 \rightarrow{} 4 \rightarrow{} d$,  he can take one of two bus lines (en-route choice) between place node $4$ and node $d$. 
	
	Discrete choice methods \cite{train2009discrete}, which is an empirical analysis of discrete choice and popular in marketing, is used to model the passenger's behavior among the set of choices (routes) $\mathcal{M}$. If there is any change in terms of the utility function $V(\vec c)$, there will be some movements among the final choices. 
	
	In the discrete choice model, there are some latent (unobserved) attributes affecting the choice of each decision-maker, such as income level and his/her preferences on driving or not driving. Consider a decision maker facing a choice among $M$ alternatives (routes). Let $v_m(\vec x)$ (defined above) be the observed utility and let $\epsilon_m$ be the unobserved utility of the $m$-th alternative, that is, $v_m(\vec x)$ does not vary from person to person, and $\epsilon_m$ is a random variable with a given distribution. From the Random Utility Models (RUMs), we have
	\beqq{
		u_m=v_m+\epsilon_m,  \text{\space}  \forall m=1,...,M
	}
	where $u_m$ is the perceived utility that the decision-maker obtains after following the $m$-th route.
	
	In order to simplify the analysis, we use logit model as it takes a closed form and is readily interpretable \cite{train2009discrete}. The unobserved utility across routes follows a type I extreme value (Gumbel) distribution. The density of $\epsilon_m$ is:
	\beqq{
		f_{\epsilon}(\epsilon_{m})= e^{-\epsilon_m} e^{-e^{-\epsilon_m}}
	}
	For a given route utility vector $\vec v=[v_1, \ldots, v_M]^\transpose$, the choice probability of route $m$ $\forall m \in [1,\ldots,M]$ is 
	\beq{ \begin{split} {\label{equ:pro}}
			p_{m}(\vec v)&=Prob(v_m+\epsilon_{m}>v_j+\epsilon_{j},\forall j \neq m) \\
			&=\int \ind{\epsilon_j < \epsilon_m+ v_m - v_j,\forall j \neq m} f_{\epsilon}(\epsilon) d\epsilon 
			= \frac{e^{v_m}}{\sum_{j=1}^M e^{v_j}}
		\end{split}
	}
	
	We also define the route choice function as $P: \Re^M \rightarrow [0,1]^M$ and $P(\vec v)=[p_1(\vec v),\ldots,p_M(\vec v)]^\transpose$, which is normalized by definition, i.e., $\sum_{m=1} ^{M} p_{m}(\vec v)=1$.

	\subsection{Elastic Demand}
	Intuitively, if the cost of this trip decreases, more people would like to enter this market and generate new trips. We use the elastic demand to capture this characteristic.
	
	Let map $S:\Re^M \rightarrow \Re$ denote the satisfactory level of passengers and $S(\vec v)= \mathbb{E}[ \max_m \{u_m\}] = \mathbb{E}[ \max_m \{v_m + \epsilon_m\}]$, where $v_m$ is the route utility of $m$-th choice  and $\epsilon_m$ is the unobserved utility.
	
	Let function $D:\Re \rightarrow \Re^+$ denote the elastic demand function. For a {given} satisfaction level $s$, $D(s)$ is the total number of passengers who enter this multi-modal transportation system.
	

	\section{Preliminaries: Stochastic User Equilibrium}
	In this section, under some assumptions about cost functions and user choice behaviors, we discuss the traffic assignment problem and leverage a fixed point framework to find the equilibrium link flow. {We mainly recall the result in \cite{cantarella1997general}}.
	
	\subsection{Traffic Assignment}
	
	Here, we state the flow calculation under a fixed link cost vector $\vec c$. 
	
	Let $\mathcal D = \{ \hat{d} \in [0, \hat{d}_{max}]\} \subset \Re^+$ denote the set of feasible demands, where $\hat{d}_{max}$ is the total population at one area. Let $\mathcal X=\{\vec x \geq 0, \sum_{i=1}^M x_m = \hat{d} \in \mathcal D\} \subset \Re^M$ be the set of feasible route flows and we have $\vec x=\vec p \hat{d}$ for a given route choice probability vector $\vec p=[p_1, \ldots, p_M]^\transpose$ and demand $\hat{d}$. We further define the set of feasible link flows as $\mathcal F=\{\vec f=B\vec x, \vec x \in \mathcal X\} \subset \Re^L$ and the set of feasible link costs  is $\mathcal C =\{\vec c=C(\vec f): \vec f \in \mathcal F\} \subset \Re^L$.

	For any given link cost vector $\vec c \in \mathcal C$, we define the traffic assignment function as $g: \mathcal C \rightarrow \mathcal F$, mapping from the link cost vector $\vec c$ to the link flow vector $\vec f$, developed from $\vec f=B\vec x$ and $\vec x=\vec p d$:
	\beq{g(\vec c)=B P(V(\vec c ))) D(S(V(\vec c)))  
	}
	where $B$ is the mapping from link to route. $P(\vec v)$ is the route choice function and $D(s)$ is the demand function of the system.  $V(\vec c)$ is the route utility function defined in Equation \ref{equ:utility} and  $S(\vec v)$ is the satisfactory function based on the route utility.

	\subsection{Fixed point formulation}
	As link cost $\vec c$ is a function of the link flow $\vec f$, the actual link flow $\vec f^*$ in the multi-modal transportation network with elastic demand will be the equilibrium flow under a certain link cost function $C$. We express it as a fixed point problem in the link flow space $\mathcal F$.

	We combine the traffic assignment mapping $g$ with the cost flow function $C (\vec f)$. For a fixed route specific cost vector $\vec c_s$, let $T: \mathcal F \rightarrow \mathcal F$ denote the composite function $g$ and $C$, which computes the new flow on each link based on the previous flow and associated cost. We have 
	\beq{\label{equ:t}
		T (\vec f) := g(C (\vec f))=B  P(V(C(\vec f) ))) D(S(V(C(\vec f) ))), \forall \vec f \in \mathcal F
	}
	
	We define the corresponding equilibrium link flow $\vec f^*$ as the fixed point of mapping $T$, i.e., $T(\vec f^*)=\vec f^*$. The following lemma concludes the existence of such an equilibrium flow {$\vec f^*$ under a certain traffic assignment function $g$.} 
	
	\begin{lemma}
		If the link flow set $\mathcal F$ is a convex compact subset of $\Re^L$, the cost-flow function $C (\vec f)$, utility function $V(\vec c)$, satisfactory function $S(\vec v)$, choice function $P(\vec v)$ and demand function $D(s)$ are continuous, then at least one equilibrium flow $\vec f^*$ exists.
	\end{lemma}
	\begin{proof}
		As the composition of continuous functions is continuous, the mapping $T$ in Equation \ref{equ:t} is continuous. Since $\mathcal F$ is convex and compact, and $T$ maps from $\mathcal F$ to itself, from the Brouwer fixed-point theorem, there exists at least one fixed point $\vec f^* \in \mathcal F$ such that $T(\vec f^*)=\vec f^*$, which is the equilibrium flow in the system.
	\end{proof}

	Under some mild assumptions (conditions 1-5 in the following lemma), we can conclude the uniqueness of the equilibrium flow $\vec f^*$, {following results in \cite{cantarella1997general}}.
	\begin{lemma}\label{lemma:unique}
		If 
		\begin{enumerate}
			\item the link cost-flow function $C (\vec f)$ is monotone non-decreasing, i.e., $(C(\vec f_1)-C(\vec f_2))^\transpose (\vec f_1-\vec f_2) \geq 0$,   $ \forall \vec f_1, \vec f_2 \in \mathcal F$,
			\item the utility function $V(\vec c)$ is linearly non-decreasing with route costs $B^\transpose \vec c$, i.e., $V(\vec c)= -\beta B^\transpose \vec c -\vec c_s$ and $\beta \geq 0$,
			\item the choice map $P(\vec v)$ is strictly positive additive, i.e., the probability density function of unobserved utility is strictly positive: $f_{\epsilon}(\epsilon_m) >0, \forall \epsilon_m \in \Re$,
			\item the choice map $P(\vec v)$ is regular, i.e., $P(\vec v_1)^\transpose (\vec v_1-\vec v_2)\geq S(\vec v_1)-S(\vec v_2) \geq P(\vec v_2)^\transpose (\vec v_1-\vec v_2), \forall \vec v_1, \vec v_2 \in \mathcal V$,
			\item the demand function $D(S(\vec v))$ is monotone non-decreasing with route utility vector $\vec v$, $(D(S(\vec v_1))-D(S(\vec v_2)))(S(\vec v_1)-S(\vec v_2)) \geq 0, \forall \vec v_1, \vec v_2 \in \mathcal V$
		\end{enumerate}
		then, there exists at most one equilibrium link flow $\vec f^*$.
	\end{lemma}
	\begin{proof}
		This lemma follows condition ii.1) and ii.2) in Theorem 2 \cite{cantarella1997general}, where condition ii.2) follows lemma 1 and Lemma 3 in \cite{cantarella1997general}.
	\end{proof}

	\subsection{Multi-class case}
	As passengers have various preferences and different sensitivity to travel time and financial cost, it is essential to consider the multi-class case, where passengers in one class share similar behavior. If the population is divided into $K$ classes, where each class of passengers has the same demand function $D^{(k)}$, satisfaction function $S^{(k)}$, choice probability $P^{(k)}$, utility function $V^{(k)}$ , route link relationship $B^{(k)}$ and route specific cost $\vec c_s^{(k)}$, then the flow on each link is the sum of the link flow of all classes, as follows.
	
	\beq{\label{equ:multiclass}
		\vec f=\sum_{k=1}^K B^{(k)}  P^{(k)}(V^{(k)}(C(\vec f)))) D^{(k)}(S^{(k)}(V^{(k)}(C(\vec f) )))
	}
	The link cost function $C(\vec f)$ is consistent across all class $k$. {Based on  results in \cite{cantarella1997general}, the following lemma concludes the uniqueness of the equilibrium flow $\vec f^*$ in the multi-class case.}
	
	\begin{lemma}
		If for each class $k$, demand function $D^{(k)}$, satisfaction function $S^{(k)}$, choice probability $P^{(k)}$, utility function $V^{(k)}$ satisfy the conditions 2 3 4 5 in Lemma \ref{lemma:unique}, and $C(\vec f)$ satisfies condition 1 in in Lemma \ref{lemma:unique},
		then, there exists at most one equilibrium link flow $\vec f^*$.
	\end{lemma}
	\begin{proof}
		This lemma follows condition ii.1) and ii.2) in Theorem 2 \cite{cantarella1997general}, where condition ii.2) follows Lemma 1 and Lemma 3 in \cite{cantarella1997general}.
	\end{proof}
	
	In order to simplify notations, we consider that passengers only have one class in the following analysis. However, these results can be easily extended to the multiclass case by following Equation \ref{equ:multiclass}. Besides, a two-class numerical simulation is developed in Section VII.

	\section{Properties of the Equilibrium Flow}
	
	In this section, we {develop} the continuity and monotonicity of the equilibrium flow, with respect to the link incentive $\vec J$ and link cost $C(\vec f)$, i.e., how the equilibrium flow changes as a result of changes in link incentives and link costs. Through this section, we show that little changes in link incentive $\vec J$ will not lead to a significant change in equilibrium flow $\vec f^*$, which offers a good opportunity for achieving higher system profit by carefully designing incentives $\vec J$.
	
	We define the incentive vector $\vec J$ as $[j_1, \ldots, j_L]^\transpose \in \Re ^L$ and the new link cost function becomes $C(\vec f) + \vec J$, where {it is an excess charge on passengers when $j_1 > 0 $ and it is price discount for link 1 when $j_1 < 0$.}
	
	\subsection{Continuity with respect to link incentive \texorpdfstring{$\vec J$}{}} 
	
	Let $\mathbb J$ denote the set of link incentives.
	We define the map $h: \mathbb J  \rightarrow \mathcal{F}$ as follows, which maps the link incentive to an equilibrium flow $\vec f^*$:
	\beq{\label{equ:h}
		h(\vec J)= \vec f^*
	}
	where $\vec f^*$ is the fixed point of mapping 
	{
		\begin{align} \label{equ: hat T}
			\hat{T} (\vec f) := g(C(\vec f)+\vec J)=B  P(V(C(\vec f)+\vec J ))) D(S(V(C(\vec f)+\vec J))), \textit{where \space}  \vec f \in \mathcal F
		\end{align}
	}
	Given functions $(C,P,S,D)$, define another map $\phi:\Re^L  \times \Re^L   \rightarrow \Re^L$ as $\phi( \vec J, \vec f)=g(C (\vec f)+\vec J)-\vec f$.
	We show how the equilibrium flow $\vec f^*$ changes as a result of the change in link incentive $\vec J$ through the following propositions. 
	
	{The first} proposition suggests that small changes in the link incentive $\vec J$ will not lead to a significant change in equilibrium flow $\vec f^*$. {The full proof is stated in \ref{sec: proof_pro_1}. }

	\begin{proposition}\label{prop:cont}
		Given incentive $\vec J_o$ and associated equilibrium flow $\vec f_o$,  if $\phi(\vec J, \cdot)$ is locally one-to-one in a neighborhood of $(\vec J_o, \vec f_o )$ (or $\phi$ is differentiable and $\nabla_{\vec f}\phi(\vec J_o, \vec f_o)$ is  invertible), then $h$ is a continuous map in a neighborhood of $\vec J_o$.
	\end{proposition}
	
	The following proposition further provides the Lipschitz continuity property of mapping $h$, which limits how fast the equilibrium flow $\vec{f}^*$ can change with respect to the link incentive $\vec J$ and will be used in the proof of our main algorithm {(Theorem \ref{thm:twotimescale}). We defer the detailed proof to \ref{sec: proof_pro_2}. }

	\begin{proposition}\label{prop:lip}
		If $g$ is a Lipschitz continuous map with Lipschitz coefficient $\beta$ and $\hat{T}$ is a contraction with parameter $\alpha$, then $h$ is Lipschitz continuous with Lipschitz coefficient $\frac{\beta}{1-\alpha}$.
	\end{proposition}

	\subsection{Monotonicity with respect to link cost}
	This subsection shows that, by reducing link cost on one specific link, the equilibrium flow on that link will increase, which may lead to higher revenue and higher profit. 
	
	For given functions  $(P,S,D)$ and $\vec J=0$, we want to discuss the effects of cost functions $C$ on the equilibrium flow $\vec f^*$. In the following proposition, we find that the link equilibrium flow $f^*_l$ is monotone decreasing with its link cost $c_l(\vec f)$.

	\begin{proposition}\label{prop:mono}
		If assumptions in Lemma \ref{lemma:unique} hold and there are two link cost functions $C_1(\vec f), C_2(\vec f)$, which only differ from each other on one link $l$, i.e.,  $c_{1,l}(\vec f) > c_{2,l}(\vec f)$, $c_{1,i}(\vec f) =c_{2,i}(\vec f), i \neq m, i \in \mathcal L, \forall  \vec f$,  Besides, $\vec f^*_1, \vec f^*_2$ are the associated equilibrium flow and $f^*_{1,l}, f^*_{2,l}$ are the equilibrium flow on link $l$, then $f^*_{1,l} < f^*_{2,l}$.
	\end{proposition}

	This proposition shows that by reducing link cost on one specific link, the equilibrium flow on that link will increase, which may lead to higher revenue and higher profit. Inspired by this idea, we introduce a link incentive $\vec J$ in the next section and want to achieve higher system profit by carefully designing $\vec J$.  {The formal proof is deferred to \ref{sec: proof_pro_3}. }

	\section{Optimal Pricing with {Passenger-friendly} Constraints}\label{sec:pricing}
	In this section, we discuss the design of link incentives $\vec J$ in order to maximize the platform's profit while maintaining passenger-friendly constraints.

	For the link cost function $C(\vec f)$, after applying the incentive on each link, we have the new link cost as $C(\vec f)+\vec J=C_p(\vec f)+C_t(\vec f)+\vec J$, where $C_p(\vec f)+\vec J$ is the price charged for passengers (the revenue of the platform {per passenger}) and $C_t(\vec f)$ is the equivalent time cost. {If $\vec J >0$, it is excess charge on passengers, while it is price discount when $\vec J <0$. i.e., passengers will give $C_p(\vec f) + \vec J$ amount of money to platform.} Assume that the link-wise supply cost for the service provider is $W(\vec f)$ where $W: \Re^L \rightarrow \Re^L$. Then the link-wise profit function without incentive is $\pi(\vec f)=C_p(\vec f)-W(\vec f)$, where $\pi: \Re^L \rightarrow \Re^L$ is the profit function. After applying incentive $\vec J$, the link-wise profit is $\pi(\vec f)+\vec J$.
	
	We want to maximize the profit of the platform under some passenger-friendly constraints, by controlling link incentive $\vec J$. The optimization problem is as follows
	\beq{\label{equ: max_profit} 
		\max_{\vec J\in\Re^L}  h^T (\vec J) \big( \pi(h(\vec J)) + \vec J \big ) }
	\beq {\label{equ: constraint} s.t. \text{\space}  B^T \vec J \leq 0, \vec J_{min} \leq \vec J \leq \vec J_{max}} 
	where the objective function is the total profit of the platform and $h(\vec J)=\vec f^*$ is defined in \eqref{equ:h}. The first  constraint {(passenger-friendly)} requires that the route incentive $B^T \vec J$ is non-positive, i.e., under the same flow, the total route cost $B^T (C(\vec f)+ \vec J)$ is less than or equal to its original route cost without any incentive ($\vec J=0$). Although the price of each route drops, at the same time, some link costs may increase. The second constraint indicates that incentive $\vec J$ is lower bounded by a constant vector $\vec J_{min}$, which prevents the link price $C_p(\vec f)+\vec J$ from becoming negative. The incentive $\vec J$ is also upper bounded by $\vec J_{max}$.

	As $h(\vec J_o)$ is the equilibrium flow under incentive $\vec J_o$, it can be calculated through iteration from the Method of Successive Averages 
	\beq{\label{equ:algo_f_jo}
		\vec f_{k+1} = (1-\alpha_k)\vec f_k +\alpha_k g(C(\vec f_k)+\vec J_o) 
	}
	where $\alpha_k$ is the decreasing step size.
	
	Traditionally, if we want to find the optimal incentive $\vec J^*$, we need to run one iteration \ref{equ:algo_f_jo} for each candidate $\vec J$, which is not feasible when the number of candidate $\vec J$ is very large or even infinite.  Inspired by the above iteration, we want to develop an algorithm, which achieves both optimal incentive $\vec J^*$ and the equilibrium flow $\vec f^*$ at the same time. 
	
	Due to the complexity of the problem, we do not know the explicit form of function $h(\vec J)$. We use the gradient of function $g$ to approximate the gradient of $h$, $\nabla h(\vec J_o) \approx \nabla g(C(\vec f_o)+\vec J_o)$. Therefore, $h(\vec J) \approx h(\vec J_o) + \nabla h(\vec J_o) (\vec J-\vec J_o)\approx \vec f_o + \nabla g(C(\vec f_o)+\vec J_o) (\vec J- \vec J_o)$.
	Here is the two time scale stochastic approximation algorithm:
	\beq{
		\vec f_{k+1} = (1-\alpha_k)\vec f_k +\alpha_k g(C(\vec f_k)+\vec J_k)\label{equ:algo_f}}
	\beq{\vec J_{k+1} = (1-\beta_k) \vec J_k + \beta_k \psi(\vec f_k, \vec J_k) \label{equ:algo_j}}
	where 
	\beq{ \begin{split} \label{equ:psi}
			\psi(\vec f_k, \vec J_k) =  \arg\max_{\vec x\in\Re^L} &\Big(\vec f_k+ \nabla g(C(\vec f_k)+\vec J_k)(\vec x- \vec J_k) \Big)^T  \\
			&\Big(\pi \big(\vec f_k+\nabla g(C(\vec f_k)+\vec J_k)(\vec x-\vec J_k) \big)+ \vec x \Big) \\
			\text{ subject to } &B^T \vec x \leq 0, \vec J_{min} \leq \vec x \leq \vec J_{max}
		\end{split}
	}
	
	Equation \ref{equ:algo_f} follows the original convergence of equilibrium flow in Equation \ref{equ:algo_f_jo} by replacing constant incentive $\vec J_o$ with incentive $\vec J_k$. Equation \ref{equ:algo_j} optimizes $\vec J$ along the direction of maximum profit where function $\psi$ solves the optimization problem based on current flow $\vec f_k$ and incentive $\vec J_k$.

	Define the flow error as $\delta \vec f= ||g(C(\vec f_k)+\vec J_k)- \vec f_k||$ and incentive error as $\delta \vec J = ||\psi(\vec f_k, \vec J_k)- \vec J_k||$. We explicitly state the algorithm in Algorithm \ref{alg:two}. From line 1 to line 5, we compute the equilibrium flow $\vec f'$ without any incentive ($\vec J_o=0$). From line 6 to line 11, the optimal incentive $\vec J^*$ and associated equilibrium flow $\vec f^*$ are calculated through the two time scale stochastic approximation algorithm.
	
	\begin{algorithm}
		\SetAlgoLined
		\SetKwInOut{Input}{Input}\SetKwInOut{Output}{Output}
		\Input{$g(\vec c)$, $\psi(\vec f, \vec J)$, $C(\vec f)$, $\alpha_k$, $\beta_k$ and $\epsilon_i$.}
		\Output{$\vec f^*, \vec J^*$}
		\BlankLine
		$\vec J_o=0$, $\vec f_1=0$, $k=1$\;
		\While{$\delta \vec f \geq \epsilon_1$}{
			$\vec f_{k+1} = (1-\alpha_k)\vec f_k +\alpha_k g(C(\vec f_k)+\vec J_o)$\;
			$k=k+1$\;
		}
		$\vec f_1=\vec f_{k}$, $\vec J_1=0$, $k=1$\;
		\While{$\delta \vec f \geq \epsilon_1$ \textbf{\textup{or}} $\delta \vec J \geq \epsilon_2$}{
			$\vec f_{k+1} = (1-\alpha_k)\vec f_k +\alpha_k g(C(\vec f_k)+\vec J_k)$\;
			$\vec J_{k+1} = (1-\beta_k) \vec J_k + \beta_k \psi(\vec f_k, \vec J_k)$\;
			$k=k+1$\;
		}
		\Return{$\vec f_k$, $\vec J_k$}
		\caption{Two time scale stochastic approximation algorithm for optimal incentive $\vec J^*$ and associated equilibrium flow $\vec f^*$}{\label{alg:two}}
	\end{algorithm}

	In order to prove the convergence of the two time scale algorithm, we need the following lemma.
	\begin{lemma}
		\label{lem:f}
		For a fixed incentive $\vec J_o$, if $\sum_{k=1}^{\infty} \alpha_k=\infty$ and $\sum_{k=1}^{\infty} \alpha_k^2 < \infty$, then $\vec f_{k+1}$ in Equation \ref{equ:algo_f_jo} converges to $\vec f^*$ and $\vec f^*=h(\vec J_o)$.
	\end{lemma}

	The above lemma concludes the convergence of equilibrium flow under a fixed incentive $\vec{J}_o$.  {The detailed proof is deferred to \ref{sec: proof_lem_4}. } In the following, we introduce an assumption about function $\psi$, which is used to prove the convergence of our algorithm.
	
	\begin{assumption}\label{asm:psi}
		We assume that
		\beqq{\Big(h(\vec J^*)-h(\vec J)\Big)^T \Big(\psi(h(\vec J),\vec J)-\vec J\Big)<0, \forall \vec J \neq \vec J^*.
		}
	\end{assumption}
	
	
	This assumption is similar to the monotonicity property assumed in Proposition \ref{prop:mono}: if the link cost is reduced on one specific link, the equilibrium flow on that link will increase, which may lead to higher profit. For example, if $h(\vec J^*)=h(\vec J)$ except on one link $l:f^*_l >f_l$, then $\Big(h(\vec J^*)-h(\vec J)\Big)^T \Big(\psi(h(\vec J),\vec J)-\vec J\Big)=(f^*_l-f_l)(j_l'-j_l)<0$, where $j_l'$ is the new incentive on link $j$. Compared with the current incentive $j_l$, the new incentive follows $j_l'<j_l$ after the optimization in function $\psi$, i.e. the link cost is reduced from $c_l(\vec f)+j_l$ to  $c_l(\vec f)+j_l'$, which leads to an increase of the equilibrium flow  on link $l$ and makes $f_l$ closer to optimal flow $f^*_l$. We believe that this assumption is not too stringent for the setting under consideration.

	
	{The following theorem concludes the convergence of Algorithm \ref{alg:two}, i.e., the output of this algorithm $(\vec f_{k+1},\vec J_{k+1})$  solves the optimization problem in Equation \ref{equ: max_profit} and \ref{equ: constraint}, which maximizes the platform’s profit while maintaining {passenger-friendly} constraints. We defer the detailed proof to \ref{sec: proof_thm_1} }
	
	\begin{theorem}\label{thm:twotimescale}
		If Assumption \ref{asm:psi} holds, $\alpha_k,\beta_k$ is tapering step size and $\beta_k/\alpha_k\rightarrow 0$, then $(\vec f_{k+1},\vec J_{k+1})$ in Algorithm \ref{alg:two} converges to $(\vec f^*,\vec J^*)$.
	\end{theorem}

	{
		\begin{remark}
			If $\vec J_{min} \leq 0 \leq \vec J_{max}$, the profit in Equation \ref{equ: max_profit}  will not decrease.
		\end{remark}
	}

	\section{Fair Profit Sharing Scheme for Providers}
	In this section, we discuss the profit sharing scheme among providers. We first introduce some basic schemes in the fair division and two-person bargaining problem. Besides, an asymmetric n-person Nash Bargaining solution is introduced and its closed-form solution is provided, which ensures each provider will have a profit increase after cooperation.
	
	\subsection{Fair division}
	As money is a divisible homogeneous resource, there are multiple principles for deciding how to divide \cite{moulin2004fair}. Let the total profit after cooperation be $R_c$ and there are $N$ providers in the system. Each provider $i$ has a utility function $U_i(R_i)$, where $\sum_{i} R_i=R_c$. Let $\vec R= [R_1, \ldots, R_N]^T$. 
	
	\textbf{Envy-free}: In this case, $R_i=\frac{R_c}{N}, \forall i$. If each provider gets the same amount of profit, then no one envies another provider, regardless of the utility function.
	
	\textbf{Egalitarian}: $\vec R=\arg\max_{\vec R}\min_{i} U_i(R_i)$, which maximizes the minimum utility among providers. If the intersection of utility function ranges is not empty, then $U_i(R_i)=U_j(R_j), \forall i,j$, i.e., the utility of every provider is equal.
	
	\textbf{Utilitarian}: $\vec R=\arg\max_{\vec R} \sum_i U_i(R_i)$, which maximizes the sum of utilities.
	
	If different providers have various utility functions towards profit, these schemes work. Their disadvantage is that they do not consider the profit before cooperation.

	\subsection{Bargaining problem}
	The two-person bargaining problem studies how two agents share a surplus that they can jointly generate. Let $F \subset \Re^2$ be the feasibility set, which contains all possible payoffs and is convex. The disagreement point $d=(d_1,d_2)$, represents the payoffs to each player when no agreement is established.

	\textbf{Proportional Solution}: $R_1= \theta R'_c$ and $R_2= (1-\theta) R'_c$ where $R'_c$ is the maximum of $R_1+R_2$ which satisfies $(R_1,R_2) \in F$ \cite{kalai1977proportional}.

	\textbf{The Nash Solution}: $\vec R =\arg\max (R_1 -d_1)(R_2-d_2) $ subject to $(R_1,R_2)\in F$. John Nash proves the solution satisfies  four axioms and maximizes the product of surplus. \cite{nash1950bargaining}.
	
	\textbf{Kalai-Smorodinsky Solution}: Let $\textit{Best}_1(F)$ and $\textit{Best}_2(F)$ be the best payoff that each player can achieve in $F$. The Kalai–Smorodinsky solution $(R_1,R_2)$ follows
	\beqq{
		\frac{R_1-d_1}{R_2-d_2}=\frac{\textit{Best}_1(F)-d_1}{\textit{Best}_2(F)-d_2}
	}
	which is the maximal point in $F$ that maintains the ratio of gains \cite{kalai1975other}.
	
	Leveraging the disagreement point in the Bargaining problem, we can make sure that each provider has a profit increase after cooperation. As there are usually more than two providers in a multi-modal network and distinct providers have different Bargaining powers, we need to extend it to the multi-player case and provide players with weights. 
	
	\subsubsection{Asymmetric Nash Bargaining solution} $\vec R =\arg\max (R_1 -d_1)^\theta (R_2-d_2)^{1-\theta} $ subject to $(R_1,R_2)\in F$ and $0<\theta<1$. This is a generalized version of Nash Bargaining solution with weight $(\theta, 1-\theta)$ on two players (page 36)\cite{muthoo1999bargaining}. 
	
	
	
	
	
	
	

	\subsection{Sharing Scheme}
	In this subsection, we introduce one fair profit-sharing scheme among service providers, which overcomes the shortcomings of the schemes mentioned above.
	
	Suppose $\mathbb{S}$ denotes the set of all service providers in the system (bus, subway, shared scooter, taxi, etc) and $\mathcal{S}=|\mathbb{S}|$. Let matrix $Z \in \{0,1\}^{L \times \mathcal{S}} $ denote the relationship between links and operators, where $z_{li}=1$ only when link $l \in \mathcal{L}$ is operated by service provider $i \in \mathbb{S}$.
	
	Before cooperation, each service provider operates individually with incentive $\vec J=0$. For service provider $i$, it has profit $t_i = \sum_{l \in \mathcal{L}} z_{li} f^\phi_l \big( \pi_l( f^\phi_l) \big)$, where $f^\phi_l$ is the equilibrium flow under cost function $C$ and $\vec J=0$. Let $\vec t = (t_1, \ldots, t_\mathcal{S})$ be the vector of disagreement payoffs, i.e., the payoff each provider gets if they fail to cooperate or do not reach agreement. 
	
	After introducing the link incentive $\vec J^*$, we have the new equilibrium flow $\vec f^*$ under the new cost function $C(\vec f)+\vec J^*$. Let $R_{c}= \vec {f}^{*\transpose} \big( \pi(\vec f^*) +\vec J^* \big )$ be the total profit of the system and we want to find a fair allocation scheme for providers.  Let $\vec R = (R_1, \ldots, R_\mathcal{S})^T$ be the payoff vector and the payoff space 
	\beqq{
		F=\left \{\vec R \in \Re^\mathcal{S} \Big| \sum_{i \in \mathbb{S}} R_i = R_c, R_i \geq 0, \forall i \in \mathbb{S} \right  \}
	}
	which is convex and compact. Let $H$ be the set of payoff vectors in $F$ not dominated by any other payoff vector in $F$ (the upper-right boundary of $F$) and $H \subset F$ is the efficiency frontier of $F$.

	The bargaining problem is then: Given $F$ and $\vec t$ and assuming that all service providers act rationally, what kind of solution $\vec R$ will bargaining eventually reach? 
	
	In order to set a fair bargain, the bargaining solution must satisfy the following axioms \cite{myerson2013game}:
	\begin{enumerate}
		\item Strong efficiency (Pareto optimality): there should be no other feasible allocation that is better than the solution for one player and not worse for other players;
		\item Individual rationality: neither player should get less than he could get in disagreement;
		\item Scale covariance: if one bargaining problem can be derived from another bargaining problem by increasing affine utility transformations, then the solution should be derived by the same transform; 
		\item Independence of irrelevant alternatives: eliminating feasible alternatives that would not have been chosen should not affect the solution;
	\end{enumerate}

	{The following theorem provides the closed-form solution to the profit-sharing scheme, which ensures each provider will have a profit increase after cooperation. The formal proof is stated in \ref{sec: proof_thm_2}. }
	
	\begin{theorem}\label{thm:sharing}
		If $\vec \theta = (\theta_1, \ldots, \theta_\mathcal{S})$ is the weight vector for providers and $\theta_i>1$, there is a unique bargaining solution $\vec R^*$, where the allocated profit for each service provider $i \in \mathbb{S}$ is
		\beq{
			R^*_i=\frac{\theta_i}{\sum_{j \in \mathbb{S}} \theta_j} \left(R_c-\sum_{j \in \mathbb{S}} t_j\right) +t_i
		}
		which means, the equilibrium payoff will converge to $\vec R^*$ by bargaining.
	\end{theorem}

	\begin{remark}
		If $\theta_i=\theta_j, \forall i,j \in \mathbb{S}$, then $\vec R^*$ is the Nash Bargaining solution. 
	\end{remark}

	$\vec R^*$ is the payoff configuration generated by a limit-totally-efficient stationary subgame perfect equilibrium from Theorem 3 \cite{okada2010nash}.

	\section{Numerical Simulation}
	In this section, we use two simulations to show the efficiency of our algorithm, especially the convergence of the traffic flow and incentive. The first simulation is based on a small network in Chengdu (China) with real-world data, and it mainly shows the process of incentive designing and profit sharing scheme among different providers. The second simulation shows that this algorithm can be implemented on a large scale network with multiple origin-destination pairs, where the data is randomized from some distributions.
	
	In order to simplify the optimization in function $\psi$, we consider the case where link profit functions are linear with respect to link flows.  The following remark shows that $\psi$ becomes quadratic programming under this assumption, which significantly simplifies the calculation and reduces the time for simulation.

		Throughout this section, we assume that the link profit function is $\pi(h(\vec J))=Q h(\vec J)+\vec \pi_0$, where $Q$ is a $L$-by-$L$ matrix. In this situation,  the function $\psi$ is the solution to the quadratic programming as follows:
		\beq{ \begin{split}
				\psi(\vec f_k, \vec J_k) &=  \arg\max_{\vec x \in\Re^L}\; \vec x^T H \vec x+ b^T \vec x \text{ subject to } B^T \vec x \leq 0, \vec J_{min} \leq \vec x \leq \vec J_{max},\\
				H=&-\nabla g^T(C(\vec f_k)+\vec J_k)(Q\nabla g(C(\vec f_k)+\vec J_k)+I),\\
				b=& -\nabla g^T(C(\vec f_k)+\vec J_k)Q \vec f_k + \nabla g^T(C(\vec f_k)+\vec J_k)Q \nabla g(C(\vec f_k)+\vec J_k)\vec J_k \\
				&-\nabla g^T(C(\vec f_k)+\vec J_k) \vec \pi_o  
				-\nabla g^T(C(\vec f_k)+\vec J_k) Q^T \vec f_k \\
				&+ \nabla g^T(C(\vec f_k)+\vec J_k) Q^T \nabla g(C(\vec f_k)+\vec J_k) \vec J_k -\vec f_k +\nabla g(C(\vec f_k)+\vec J_k) \vec J_k.
			\end{split}
		}

	\subsection{Small network with profit sharing}
	In this subsection, we use a simple example in Chengdu, China, with multi-class passengers to show the efficiency of our algorithm through the convergence of the traffic flow and link incentive. The associated profit sharing scheme is also discussed. To be specific, the simulation is based on the paths demonstrated in Figure \ref{fig:hyperpath}, containing 6 nodes, 9 routes (hyperpaths), and 12 edges. 
	
	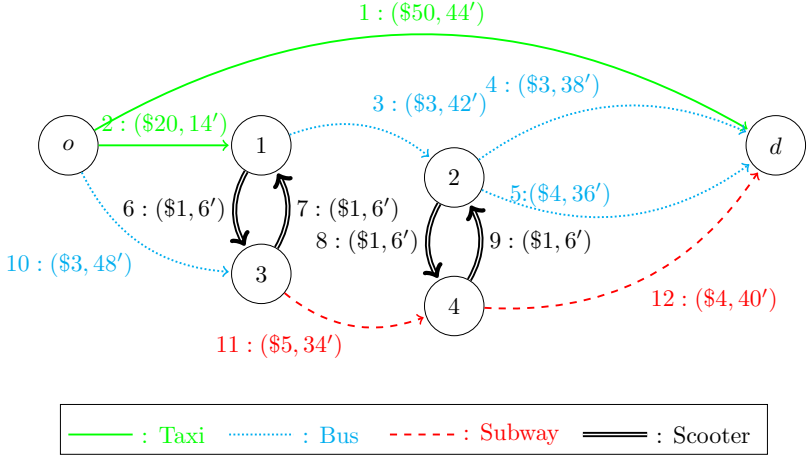
\begin{figure}[bth]
		\centering
		\scalebox{0.9}{\input{hyperpath_fig3}}
		\caption{\label{fig:hyperpath2} Multi-modal Transportation Network.}
	\end{figure}

	All nine routes (hyperpaths) are shown below: 
	\begin{enumerate}
		\item $o \to d$
		\item $o \to 1 \to 2 \to d$
		\item $o \to 1 \to 3 \to 4 \to d$
		\item $o \to 1 \to 3 \to 4 \to 2 \to d$
		\item $o \to 1 \to 2 \to 4 \to d$
		\item $o \to 3 \to 1 \to 2 \to d$
		\item $o \to 3 \to 1 \to 2 \to 4 \to d$
		\item $o \to 3 \to 4 \to 2 \to d$
		\item $o \to 3 \to 4 \to d$
	\end{enumerate}
	where the en-route choice probability from node $2$ to node  $d$ is 0.4 and 0.6 for link 4 and link 5, respectively.
	
	We assume that there are two classes of passengers. 
	\begin{itemize}
		\item Group A: Passengers can accept routes with up to three links, i.e., route 1, 2, 9.
		\item Group B: Passengers can accept all nine routes listed above.
	\end{itemize}
	For example, most senior citizens will be in group A as it is not easy for them to ride a scooter on route 6 to 9. Each group of passengers can have its own demand function, choice model, and route utility function.
	
	For each link, we define the link cost $c_i(\vec f)=c_{i,p}(\vec f)+ c_{i,t}(\vec f)$, which contains both the price and the travel time on that link. To simplify the problem, we let $c_{i,p}(\vec f)=c_{i,price}+J_i$ and $c_{i,t}(\vec f)=\gamma (0.02 f_i+ c_{i.time})$ where  $\gamma=0.5$ and travel time is increasing with link flow $f$. The actual $c_{i,price}$ and $c_{i,time}$ are listed in Figure \ref{fig:hyperpath2} and they are captured through Baidu Maps and \cite{srinivasan2020built}, where these nodes are Railway Station (origin), transfer hubs in Downtown and  Chengdu Airport (Destination) respectively, in Chengdu, China.

		

	We use the logit choice model defined in Equation \ref{equ:pro} for both groups. The utility function is $V(\vec c)=v_0-B^T \vec c$ where $v_0=200$ is a constant $\vec c= C(\vec f)$ is the link cost function. Although the single payment of the entire trip through the platform will increase the route utilities, we do not consider it here. We will show that the platform is attractive even without the single payment benefit.
	The satisfactory function is $S(\vec v)=\max_i v_i$ for both groups.  The demand functions are different for groups with different elasticity. For group A, the demand function is $D_a(s_a)=60 \tanh(s_a)$ while for group B it is $D_b(s_b)=40 \tanh(s_b)$. 
	
	We further define the profit function per person on each link. As noted before Equation \ref{equ:psi}, we assume that the link profit function is linear with respect to link flow, where $\pi(\vec f)=Q\vec f+ \vec \pi_o$, where
	\beqq{
		\begin{split}
			Q &=diag[-0.2,-0.2,0.05,0.05,0.05,-0.03,-0.03,-0.03,-0.03,0.05,0.05,0.05],\\
			\vec \pi_o & =[10,4,0.5,0.5,0.75,0.7,0.7,0.7,0.7, 0.5, 2, 1.6]^T.  
		\end{split}
	}
	This is because taxi and scooter companies need to rebalance when the flow on one link becomes large while the operating cost for bus and subway is fixed.

	\subsubsection{No incentive}
	If $\vec J^\phi=0$, we can compute the equilibrium flow through the first five lines in Algorithm \ref{alg:two} and Equation \ref{equ:multiclass}. As shown in Table \ref{table:j2}, the equilibrium flow $\vec f^\phi$, link cost $C(\vec f^\phi)$ and link profit $\pi(\vec f^\phi)$ are listed. The total profit in the system is {\$}230.34. For passengers in group A, the route flow vector is $\vec x^\phi_a=[19.58,7.30,0,0,0,0,0,0,6.94]^T$ with total demand $d^\phi_a=33.82$. For passengers in group B, route flow vector is $\vec x^\phi_b=[12.58,4.69,0.08,0.01,0.01,0.08,0.00,0.63,$ $4.46]^T$ with demand $d^\phi_b=22.55$. Most passengers choose route 1 with a direct taxi to the destination.

	\subsubsection{Add incentive}
	In this subsection, we assume that taxi, bus, subway, and scooter companies cooperate and form an all-in-one platform, which can add incentives on each link and provide passengers multi-modal services with one payment per trip. 
	
	We further set the minimum incentive on each link as $J_{min}=-3$, which prevents the cost on each link from being negative. The maximum incentive $J_{max}=3$ for each link {to make the feasible area include the global optimal solution}.   Following Algorithm \ref{alg:two} and Equation \ref{equ:multiclass} with $\alpha_k=\frac{1}{10+0.001k^{0.8}}$ and $\beta_k=\frac{1}{10+k^{0.9}}$, we have the optimal incentive $\vec J^*$, equilibrium flow $\vec f^*$, link cost $C(\vec f^*)+\vec J^*$ and link profit $\pi(\vec f^*)+\vec J^*$ in Table \ref{table:j2}. The outcome $\vec J^*$ satisfies the constraint $B^\transpose \vec J \leq 0$ in Inequality \eqref{equ: constraint}, which makes the cost on each route less than or equal to the original route cost before incentivizing.
	
	\begin{table}[h!]
		\centering
		\begin{tabular}{|m{3em} | c c c | c c c c |}
			\hline
			Link &$\vec f^\phi$ &$C(\vec f^\phi)$ &$\pi(\vec f^\phi)$ &$\vec J^*$  &$\vec f^*$ &$C(\vec f^*)+\vec J^*$ &$\pi(\vec f^*)+\vec J^*$  \\
			\hline
			1 &32.16 &72.32 &3.57       &-0.00  &5.15   &72.05  &8.97\\
			2 &12.10 &27.12 &1.58       &-0.35  &2.11   &26.68  &3.23\\
			\hline
			3 &12.09 &24.12 &1.10       &0.16   &1.90   &24.17  &0.75\\
			4 &5.09 &22.05  &0.75       &0.32   &0.80   &22.33  &0.86\\
			5 &7.63 &22.08  &1.13       &0.10   &1.21   &22.11  &0.91\\
			\hline
			6 &0.09 &4.00   &0.70       &-0.23  &0.22   &3.77   &0.46\\
			7 &0.09 &4.00   &0.70       &1.23   &0.01   &5.23   &1.93\\
			8 &0.01 &4.00   &0.70       &2.04   &0.00   &6.04   &2.74\\
			9 &0.64 &4.00   &0.68       &1.69   &0.11   &5.69   &2.38\\
			\hline
			10 &12.13 &27.12 &1.11      &-1.58  &49.98  &25.92  &1.42\\
			11 &12.13 &22.12 &2.61      &-1.30  &50.19  &21.20  &3.21\\
			12 &11.50 &24.11 &2.17      &-1.85  &50.08  &22.65  &2.26\\
			\hline
		\end{tabular}
		\caption{Equilibrium flow, link cost {(\$)} and link profit{(\$)} before/after incentive} \label{table:j2}
	\end{table}
	
	The total profit after cooperation in the system is {\$}401.90, which is an increase of 74\% from original profit. For passengers in group A, the route flow vector is $\vec x^*_a=[3.11, 1.14, 0, 0, 0, 0, 0, 0, 30.09]^T$ with total demand $d^*_a=34.34$. For passengers in group B, route flow vector is $\vec x^*_b=[2.04, 0.75, 0.22, 0.00, 0.00, 0.01, 0.00,$ $ 0.11,19.77]^T$ with demand $d^*_b=22.90$. Most of passengers now choose route 9.
	
	{The profit will increase for small $J_{max}$ as well, e.g. when $J_{max} = 0.1$ and $J_{min}= -0.1$, the profit increase from \$230.34 to \$246.64, with link incentive $\vec J = [0, -0.01, 0.1, -0.1, -0.1, -0.1, -0.1, -0.1, -0.1, -0.1, -0.1, -0.1 ]$ }
	
	\subsubsection{Comparison and Convergence}
	Comparing two cases with and without incentives in Table \ref{table:j2}, we find that 1) Some passengers shift from pure taxi service (route 1) to multi-modal routes and more passengers use bus and subway services.  As the difference between route flow vector is $\vec x^*_a+\vec x^*_b-\vec x^\phi_a-\vec x^\phi_b=[-27.00, -10.10, 0.14, -0.01, -0.01, -0.01, 0.00,$ $ -0.53, 38.45]^T$, less passengers choose route 1 or 2 and more choose routes 9, as a result of incentives $\vec J^*$. 2) The system profit increases from {\$}230.34 to {\$}401.90, because of both demand shifting and induced demand (more passengers enter the market). 3) Although some link costs increase, the cost of each route decreases as both $B_a^T \vec J^* \leq 0$ and $B_b^T \vec J^* \leq 0$ hold, which is defined as fairness constraint before.

	Define the flow error as $\delta \vec f= ||g(C(\vec f_k)+\vec J_k)- \vec f_k||_2$ and the incentive error as $\delta \vec J = ||\psi(\vec f_k, \vec J_k)- \vec J_k||_2$. We show the convergence of our algorithm as both errors go to zero with iteration, as shown in Figure \ref{fig:error_f} and \ref{fig:error_j}. 
	The profit of the platform increases with iteration, as shown in Figure \ref{fig:profit}. 
	
	\begin{figure}[!ht]
		\centering
		\scalebox{0.16}{\includegraphics{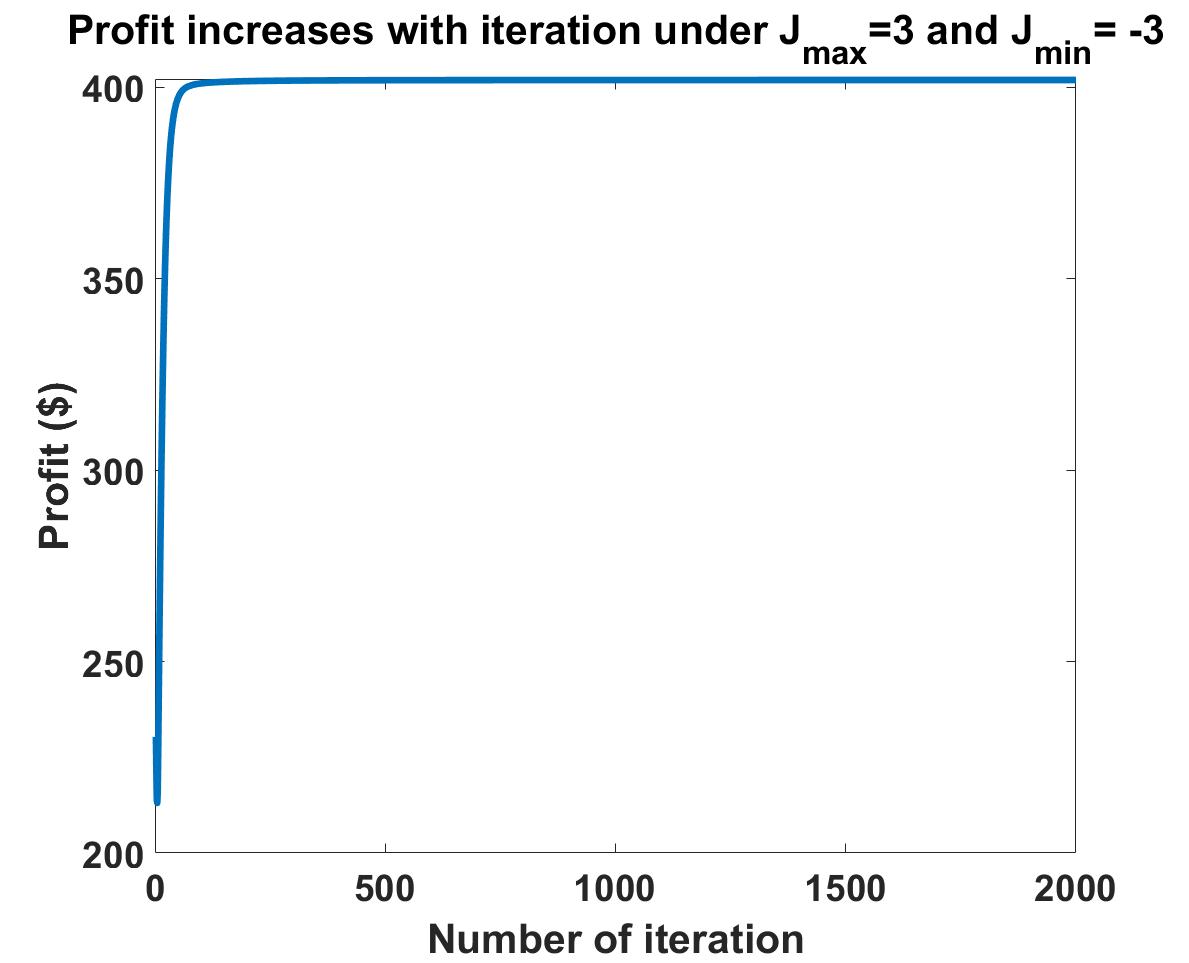}}
		\caption{Profit of the platform increases with iteration in Algorithm \ref{alg:two}}\label{fig:profit}
	\end{figure} 
	\begin{figure}[!ht]
		\centering
		\scalebox{0.17}{\includegraphics{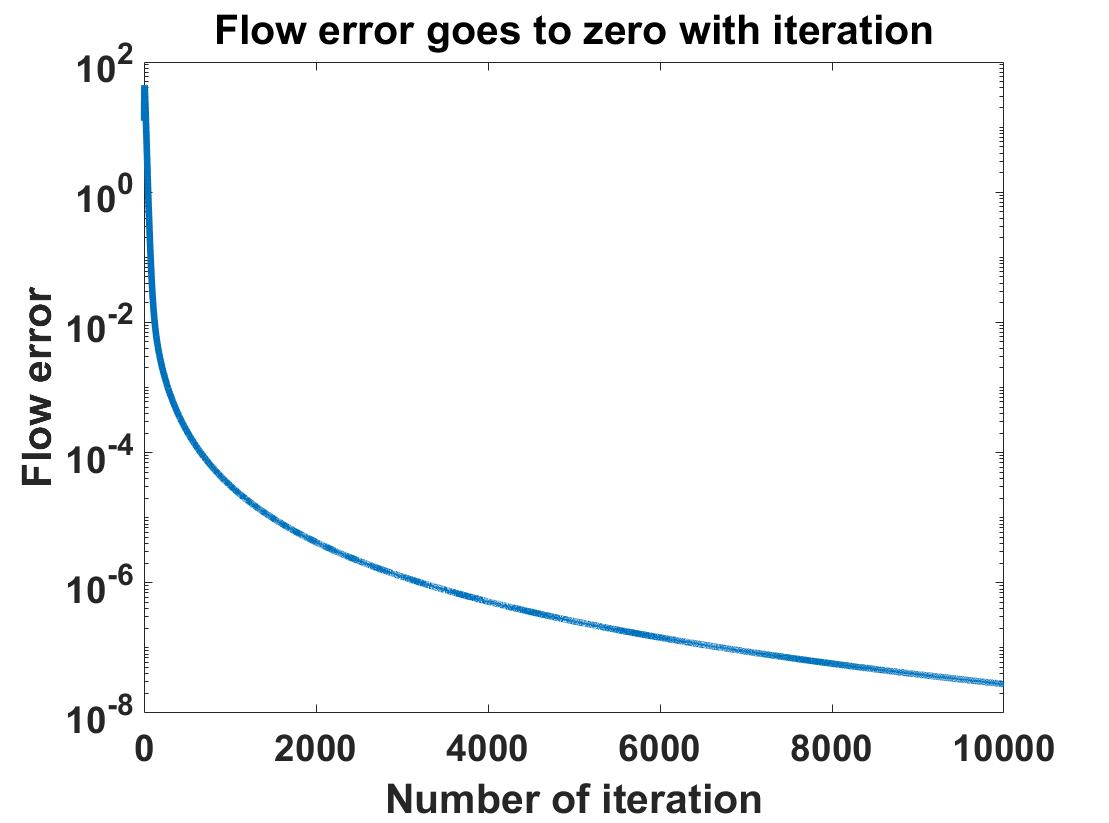}}
		\caption{Flow error $\delta \vec f= ||g(C(\vec f_k)+\vec J_k)- \vec f_k||_2$ goes to zero with iteration in Algorithm \ref{alg:two} and thus the link flow converges.}\label{fig:error_f}
	\end{figure} 
	\begin{figure}[!ht]
		\centering
		\scalebox{0.16}{\includegraphics{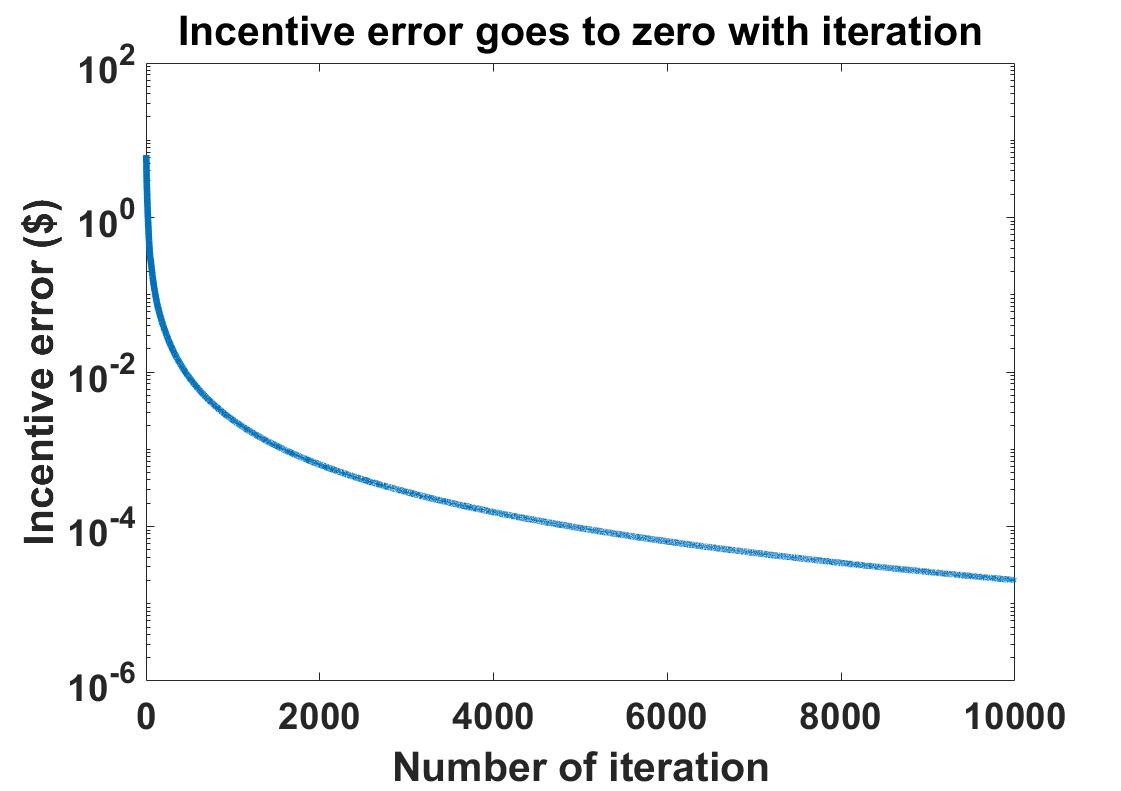}}
		\caption{Incentive error $\delta \vec J = ||\psi(\vec f_k, \vec J_k)- \vec J_k||_2$ goes to zero with iteration in Algorithm \ref{alg:two} and thus the incentive converges. }\label{fig:error_j}
	\end{figure}

	\subsubsection{Profit sharing with Asymmetric Nash Bargaining solution}
	In this subsection, we discuss the profit sharing scheme among the four providers. We assume that the weights on taxi, bus, scooter, and subway companies are $\vec \theta=[70, 60, 1, 200]$, which is set by the platform organizer. According to theorem \ref{thm:sharing}, there is a unique bargaining solution $\vec R^*=[170.15,70.35,1.08,160.32]$ out of the total profit {\$}401.90.
	
	\begin{table}[!ht]
		\centering
		\begin{tabular}{l l l l l l l}
			\hline
			&Provider &\shortstack{Profit \\(before)} & \shortstack{Profit \\(after)} &Compensation &Final profit &Increase \\
			\hline
			&taxi &133.87 &53.02 &+117.13 &170.15 &+36.28\\
			&bus &39.25 &74.31 &-3.96 &70.35 &+31.10\\
			&scooter &0.57 &0.36 &+0.72 &1.08 &+0.51\\
			&subway &56.65 &274.21 &-113.89 &160.32 &+103.67\\
			\hline
			&Total &230.34 &401.90 &0 &401.90 &+171.56\\
			\hline
		\end{tabular}
		\caption{Profit Sharing among providers via Asymmetric Nash Bargaining solution}\label{Table:sharing}
	\end{table}
	
	After applying incentive $\vec J^*$, more passengers join multi-modal routes and there is a profit increase for bus and subway. Although taxi and scooter experience profit declines, their final profits will still increase as they receive compensations from other providers. This profit sharing scheme ensures that each provider receives profit growth, which attracts more providers to join the platform.
	
	\subsection{Large scale network with multiple origin-destination pairs}
	In this subsection, we test our algorithm on a randomly generated large scale network with multiple random origin-destination pairs, which consists of 500 nodes, 1988 links, and 100 O-D pairs with elastic demand.
	
	\subsubsection{Graph generating: Barabasi-Albert Model}
	We generate a scale-free network with 500 nodes based on Barabasi-Albert Model \cite{barabasi1999emergence}. The preferential attachment mechanism is based on the idea that the more connected a node is, the more likely it is to receive new links, which is similar to a hub in transportation networks. We assume that each new node has two links to previous nodes on average.
	After generation, the graph has 1988 directed links. The constant cost of each link is generated uniformly from 10 to 20. As shown in Figure \ref{fig:power}, the distribution of degrees in the generated graph follows a power law, thus this graph is scale-free. The generated directed graph is shown in Figure \ref{fig:visual}.
	
	
	\begin{figure}[!ht]
		\centering
		\scalebox{0.235}{\includegraphics{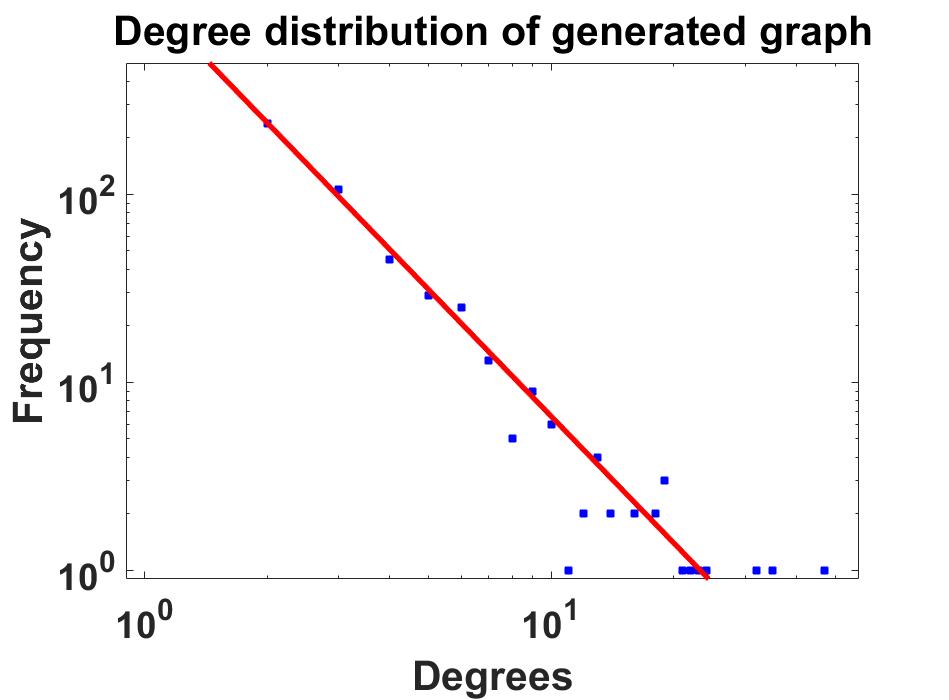}}
		\caption{The degree distribution of scale free graph, where $P(k) \sim k^{-3}$ follows a power law, i.e., $log(P)$ is linear with respect to $log (\textit{Degree})$.}
		\label{fig:power}
	\end{figure} 
	
	\begin{figure}[!ht]
		\centering
		\scalebox{0.3}{\includegraphics{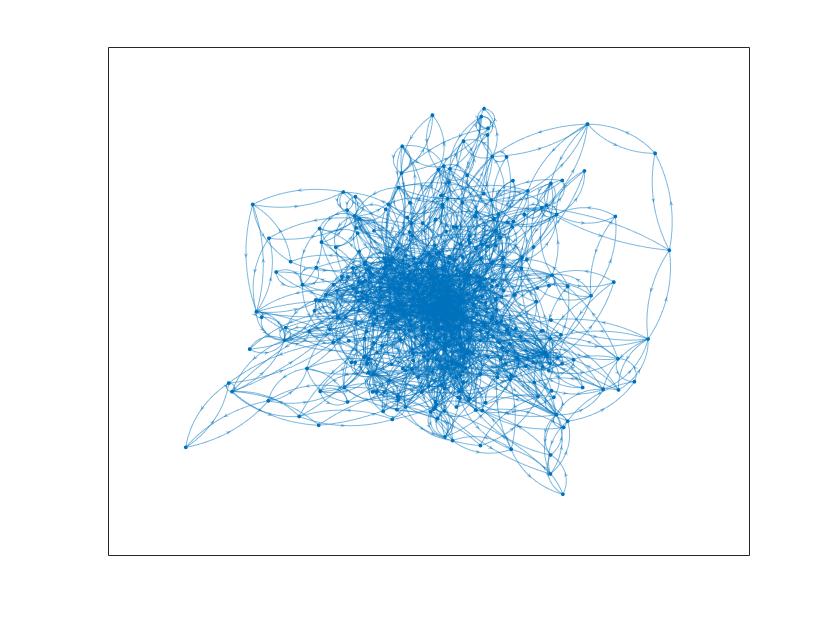}}
		\caption{Visualization of the generated network with 500 nodes and 1988 links}
		\label{fig:visual}
	\end{figure} 
	
	\subsubsection{Hyperpath generating: K-shortest path}
	We generate 100 Origin-Destination (OD) pairs where both origins and destinations are randomly selected from a uniform distribution from 1 to 100. For each OD pair, we use the K-shortest path algorithm to generate 3 hyperpaths, i.e., the shortest 3 hyperpaths from origin to destination. We do not consider other paths because customers will choose high-cost hyperpaths with an extremely low probability. We store the graph and route information in a 1988-3-100 matrix, which indicates the relationship between links, routes, and OD pairs.

	\subsubsection{Two time-scale algorithm: multiple Origin-Destination pairs}
	For each OD pair $n$, we have demand function $D^{(n)}$, satisfaction function $S^{(n)}$, choice probability $P^{(n)}$, utility function $V^{(n)}$ and  route link relationship $B^{(n)}$. In order to simplify the simulation, for each OD pair $n$, we choose $D^{(n)}=10*a_n tanh(b_n s_n)$, $S_n=max(V_n/200)$, Logit choice model, $V_n=V_0-B_n^\transpose C(\vec f)$, where $a_n$ and $b_n$ are uniformly distributed in interval $[0.9,1.1]$, $V_0 = 200$, $B_n$ (link-route relationship for OD pair $n$) comes from the 1988-3-100 matrix defined above.
	
	Besides, the actual link flow is the sum of flow under each ODpair on that link,i.e., $\vec f=  \sum_{n=1}^{100} B^{(n)}  P^{(n)} $ $(V^{(n}(C(\vec f)))) D^{(n)}(S^{(n)}(V^{(n)}(C(\vec f) )))$. In order to run Algorithm \ref{alg:two}, we only need to replace the function $g(C(\vec f_k + \vec J_k))$with $ \sum_{n=1}^{100} g_n(C(\vec f_k + \vec J_k))$ where $g_n$ is the traffic assignment function for ODpair $n$. Therefore, $\nabla g(\vec J_k)$ needs to be replaced by $\sum_{n=1}^{100}\nabla g_n(\vec J_k)$. As different ODpairs share the same link profit functions, link cost functions and link flows, other parameters in function $\psi$ remain the same.
	
	For the link cost function $C(\vec f)=Q_c \vec f + \vec c_o $, { we set $Q_c = diag(q_c)$ to make cost increasing with respect to flow, where $q_c \in \{0.005, 0.01, 0.015\}$ denotes three different modes.} The link constant cost $\vec c_o$ is uniformly distributed in interval $[10,20]$.  For link profit function $\pi(\vec f)=Q\vec f+ \vec \pi_o$, we set  $\vec \pi_o=\vec c_o/2 $  and $Q$ is diagonal matrix where each element follows uniform distribution between -0.1 and 0.1, in order to simplify the simulation.
	
	\subsubsection{Comparison and Convergence}
	
	By setting stepsize as $\alpha_k=\frac{1}{10+0.001k^{0.8}}$ and $\beta_k=\frac{1}{100+0.8k^{0.9}}$, Algorithm \ref{alg:two} converges to the optimal $J^*_c$. As there are only 100 OD pairs, 577 out of 1988 links have nonzero link flow. In the following, we compare the incentives and flows on these 577 links without including zero-flow links.
	
	As shown in Figure \ref{fig:hist_j}, the link incentives $J^*_c$ range from -3 to 3, which is restricted by $J_{max}=3$ and $J_{min}=-3$. When $J$ is negative, it reduces the cost and leads to a higher flow on this link. On the other side, it becomes a toll and reduces flow when $J$ is positive. 
	
	\begin{figure}[!ht]
		\centering
		\scalebox{0.22}{\includegraphics{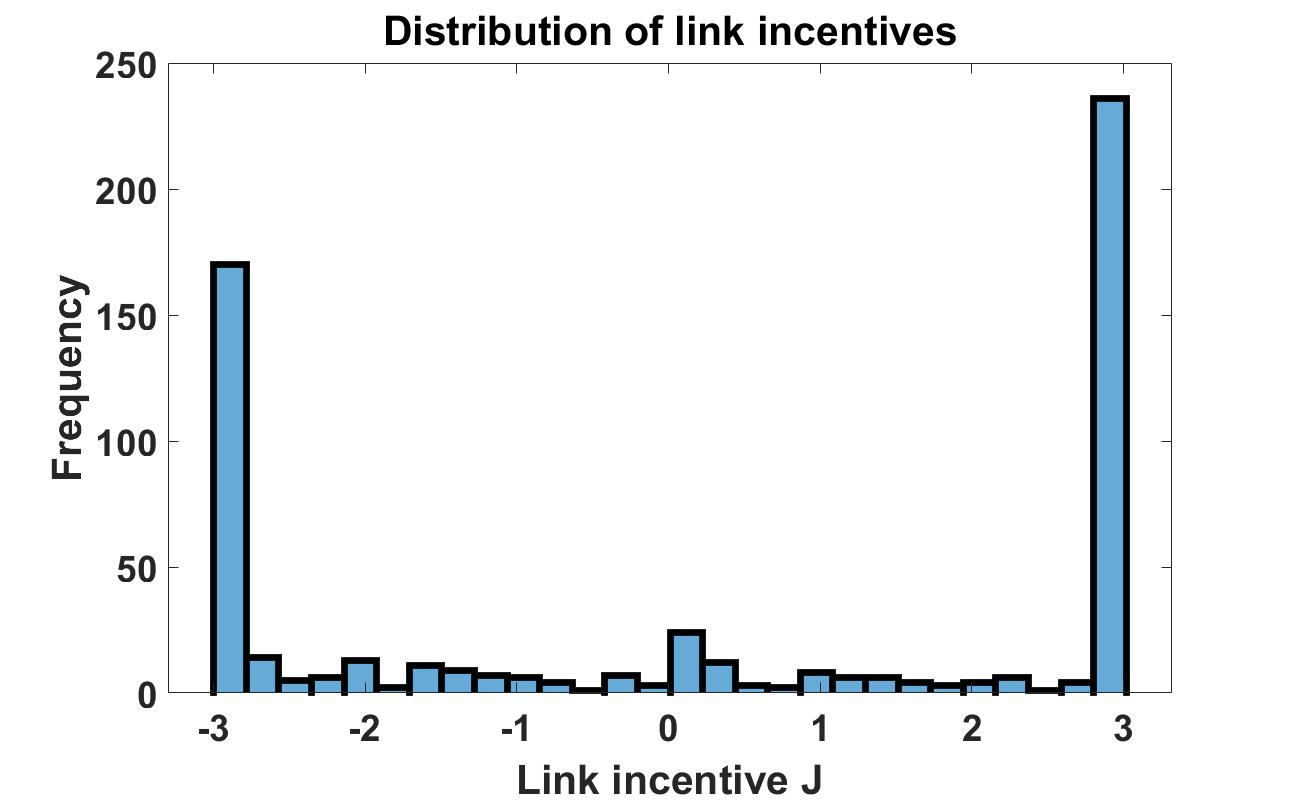}}
		\caption{Distribution of optimal link incentives $J^*$}\label{fig:hist_j}
	\end{figure}

	As shown in Figure \ref{fig:hist_f_change}, the difference of link flow is more likely to be positive, which indicates that the total demand increases after incentivizing. The change of flow mainly lies in the interval $[-2,3]$, which will not lead to a dramatic change of system and our linear approximations hold.

	\begin{figure}[!ht]
		\centering
		\scalebox{0.24}{\includegraphics{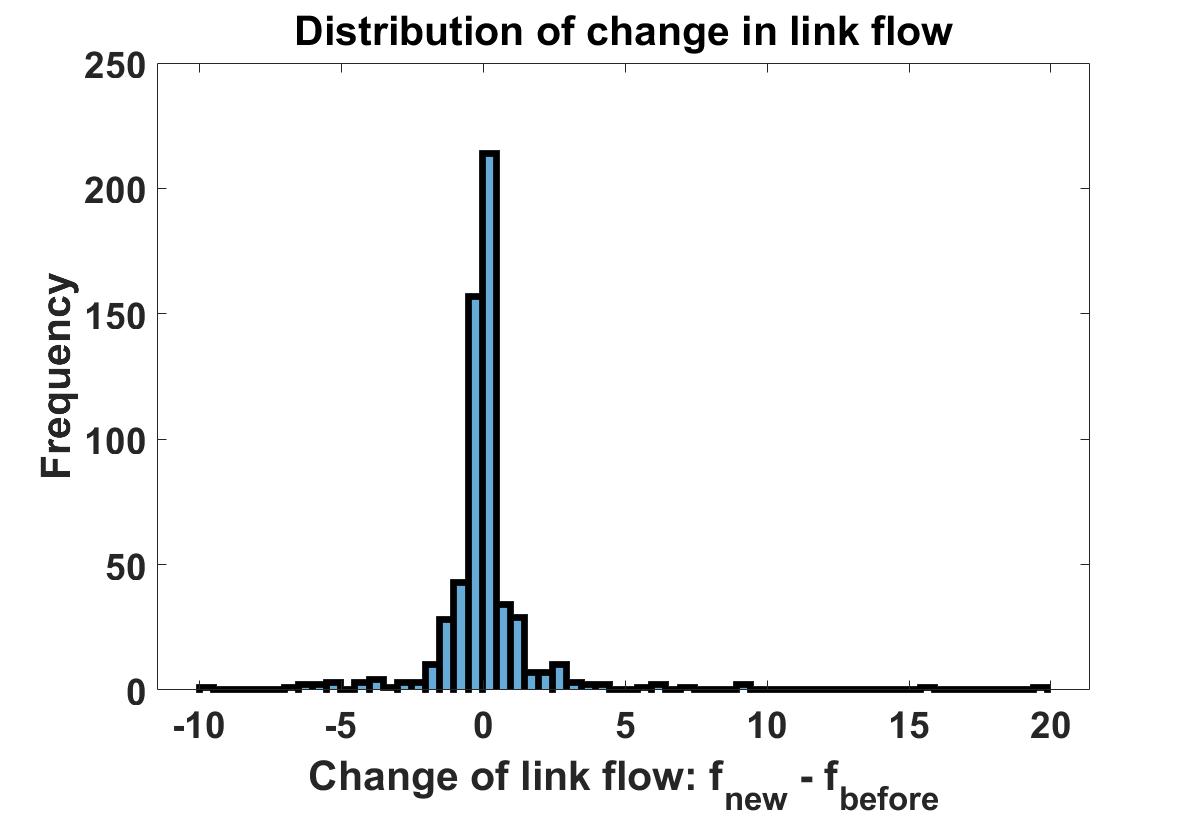}}
		\caption{The distribution of changes in link flow}\label{fig:hist_f_change}
	\end{figure}

	Similar to the small network of Chengdu, the profit of the platform increases with iteration from 17520 to 20296, as shown in Figure \ref{fig:profit_graph}. The convergence of our algorithm for multiple O-D pairs large scale network can be validated since both flow error and incentive error go to zero with iteration, as shown in Figure \ref{fig:error_f_graph} and \ref{fig:error_j_graph}. Similarly, under the Asymmetric Nash Bargaining solution in Theorem \ref{thm:sharing}, each provider can have a profit increase, which is omitted here.

	\begin{figure}[!ht]
		\centering
		\scalebox{0.22}{\includegraphics{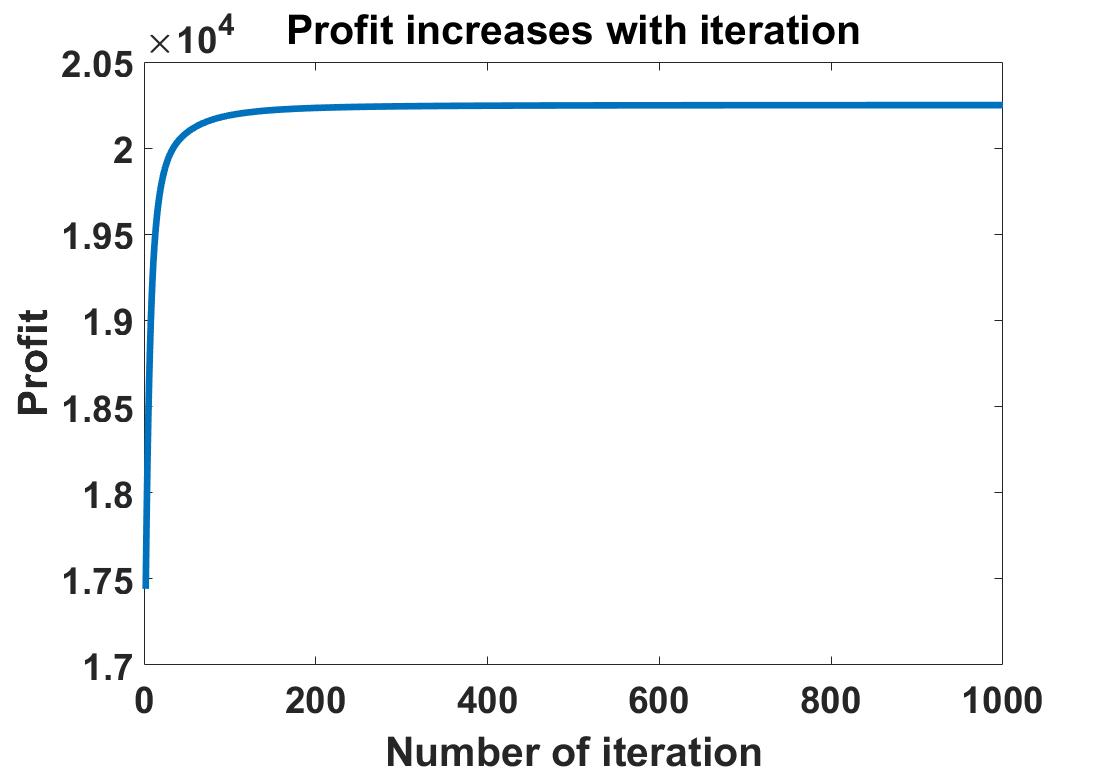}}
		\caption{Profit of the platform increases with iteration in Algorithm \ref{alg:two}}\label{fig:profit_graph}
	\end{figure} 
	\begin{figure}[!ht]
		\centering
		\scalebox{0.2}{\includegraphics{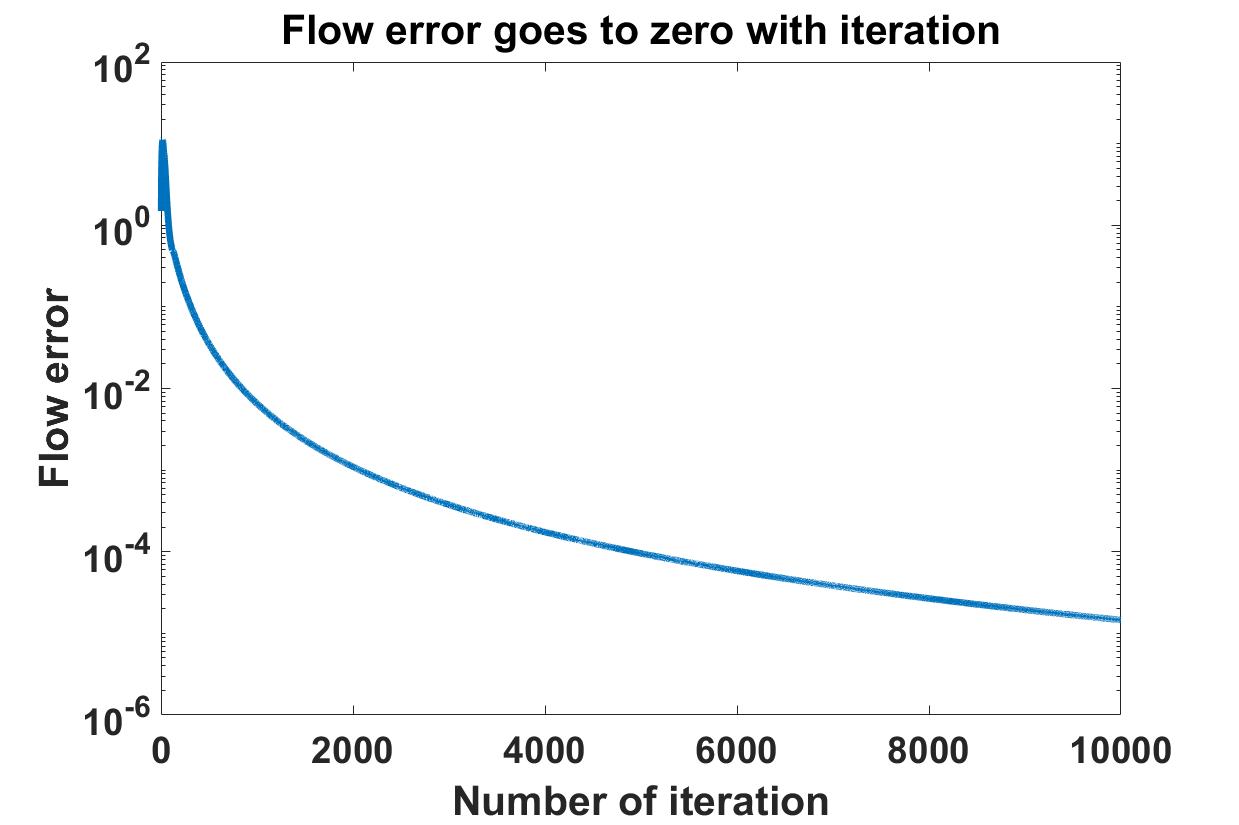}}
		\caption{Flow error $\delta \vec f= ||g(C(\vec f_k)+\vec J_k)- \vec f_k||_2$ goes to zero with iteration in Algorithm \ref{alg:two} and thus the link flow converges.}\label{fig:error_f_graph}
	\end{figure} 
	\begin{figure}[!ht]
		\centering
		\scalebox{0.2}{\includegraphics{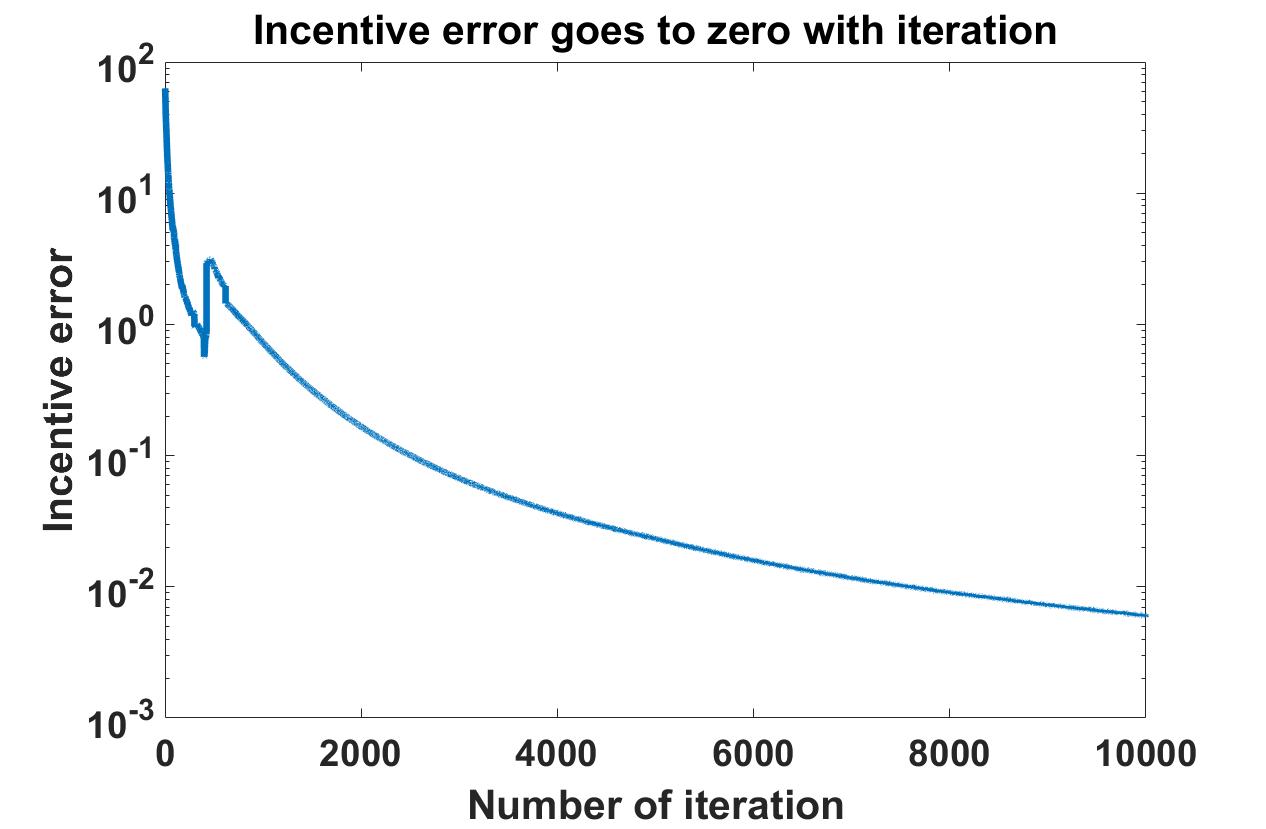}}
		\caption{Incentive error $\delta \vec J = ||\psi(\vec f_k, \vec J_k)- \vec J_k||_2$ goes to zero with iteration in Algorithm \ref{alg:two} and thus the incentive converges. }\label{fig:error_j_graph}
	\end{figure}

	{
		\subsection{Algorithm efficiency}
		We run all the simulation based on a traditional PC with an Intel i7-7700 CPU. The fist simulation in Section 7.1 takes 25 seconds and the second simulation in Section 7.2 takes 8 hours since it has 1988 links and 100 OD pairs. In addition, if we change parameters of cost function, the running time does not change. If the topology changes, say from 12 links in first simulation to 1988 links in second simulation, the run time increases significantly.
	}
	
	\section{Conclusion}
	In this paper, we present a theoretical framework to analyze the market aspect of cooperation between various competing modes of transportation. Specifically, we model a multi-modal transportation network where various service providers are encouraged to cooperate. Through an all-in-one platform, one payment for the entire multi-modal trip is simplified and the prices can be adjusted for each route depending on the demand. Under the multi-class elastic demand assumption, we use hyperpath to model the passengers' assignments and determine the equilibrium flow using the fixed point formulation presented in \cite{cantarella1997general}. After introducing incentives to encourage cooperation between various service providers and reduce the price for each route, a two time-scale stochastic approximation algorithm is developed for designing incentives and computing the associated equilibrium flow.  Based on an asymmetric Nash Bargaining Solution, we design a fair profit-sharing scheme for competing providers. Numerical simulations are provided to show the effectiveness of our algorithm on a large-scale network with multiple origin-destination pairs. Furthermore, we would like to test our framework and algorithm in real-world large-scale simulators.

	\section*{Acknowledgment}
	This research is supported by the Ford Motor Company under the University Alliance Project.

	\bibliography{sample_electric}
	
	\newpage
	
	\appendix
	
	
	{
		
		\section{List of notations} \label{sec: notations}
		In this section we provide the full list of notations.

		\subsection{Link, Path, Route}
		\begin{itemize}
			\item Link: $\mathcal{L} = \{l_1, l_2, \ldots, l_L\}$, $L =|\mathcal{L}|$
			\item Link flow: $\vec f = [f_1, \ldots, f_L]^T$, $\vec f = B\vec x$
			\item Link cost: $\vec c = C(\vec f)= [c_1(\vec f), \ldots, c_L(\vec f)]^T$, after incentivizing: $C(\vec f) + \vec J$
			\item Path: $n=(\mathcal Z_n, \mathcal L_n) \in \mathcal N$,  $N= |\mathcal N|$
			\item Link-path incidence matrix: $A=[a_{ln}] \in \{0,1 \}^{L*N}$, $a_{ln}=1$ when link $l$ belongs to path $n$ 
			\item Route (hyperpath): $m = (\mathcal Z_m, \mathcal L_m, \mathcal \pi_m), \mathcal \pi_m \in [0,1]^{|\mathcal L_M|}$
			\item Path-route incidence matrix: $E = [e_{nm}] \in \{0.1 \}^{N * M}$
			\item Link-route incidence matrix:$B =AE = [b_{lm}]\in \{0.1 \}^{L * M}$, $b_{lm}=\sum_n a_{ln} e_{nm}$
			\item Route flow: $\vec x = [x_1, \ldots, x_M] \in \Re^M$, $\vec x = \vec p d$
			\item Route utility: $\vec v = V(\vec c) = -\beta B^T \vec c - \vec c_s$
		\end{itemize}
		
		\subsection{Route probability, Demand, Incentive}
		\begin{itemize}
			\item Route choice probability: $\vec p= P(\vec V) = [P_1(\vec v), \ldots, P_M(\vec v)]^T$
			\item Satisfactory function: $s = S(\vec v) = E [\max_m \{V_m + \epsilon_m \}]$
			\item Demand function: $d = D(s)$
			\item Traffic assignment (link flow): $g(\vec c) = BP(\vec V(\vec c)) D(S(V(\vec c)))$
			\item Link incentive (charges to customers): $ \vec J = [j_1, \ldots, j_L]^T \in \Re^L$, $j_l <0$ is a discount for customers and $j_l >0$ is a toll for customers 
			\item Mapping from link incentive to link equilibrium flow: $h(\vec J) = \vec f^*$, where $\vec f^* = g( C(\vec f^*) + \vec J)$ is the equilibrium flow
			\item Link profit after incentivizing: $\pi(\vec f)+ \vec J$
		\end{itemize}
		
		\subsection{Profit Sharing}
		\begin{itemize}
			\item Prior profit(disagreement payoff): $\vec t=(t_1, \ldots, t_S)$
			\item Total profit after link incentivizing: $R_c = \vec f^{*T} (\pi(\vec f^*) +\vec J^* )$ 
			\item Profit after incentivizing: $\vec R = [R_1, \ldots, R_S]^T \in \Re^S$
			\item Weight vector for service providers: $\vec \theta = (\theta_1, \ldots, \theta_S)$
		\end{itemize}
	}

	\section{Proofs in Section 4} \label{sec: proof_sec_4}

	\subsection{Proof of Proposition \ref{prop:cont}}  \label{sec: proof_pro_1}
	
	\begin{proof}
		Since $\vec f_o \in \Re^L$ is the fixed point of $\hat{T}$, we have $\phi(\vec J_o, \vec f_o)=0$. As $\phi$ is locally one-to-one, its partial derivative with respect to $\vec f$ is invertible. By the Implicit function theorem \cite{kumagai1980implicit}, there exists an open set $U \in \Re^L$ containing
		$\vec J_o$ and there exists a continuous function $h: U \rightarrow \Re^L$
		such that $h(\vec J_o) = \vec f_o$ and $\phi(\vec J, h(\vec J))=0, \forall \vec J \in U$.
		i.e., for a new link incentive $\vec J' \in U$, the continuous map $h(\vec J')$
		maps into the new equilibrium flow $\vec{f}^{'}$ since $\phi(\vec J',h(\vec J'))=0$.
	\end{proof}
	
	\subsection{Proof of Proposition \ref{prop:lip}}  \label{sec: proof_pro_2}
	\begin{proof}
		Let $\hat{T}_1(\vec f)=g(C(\vec f)+\vec J_1)$ and $\hat{T}_2(\vec f)=g(C(\vec f)+\vec J_2)$. Define $\vec f^*_1$ and $\vec f^*_2$ as the fixed point of $\hat{T}_1$ and $\hat{T}_2$ respectively.
		We have
		\beq{ \begin{split}
				||\vec f^*_1-\vec f^*_2||_\mathcal{F}&=||\hat{T}_1(\vec f^*_1)-\hat{T}_2(\vec f^*_2)||_\mathcal{F} \\
				&\leq ||\hat{T}_1(\vec f^*_1)-\hat{T}_1(\vec f^*_2)||_\mathcal{F}+||\hat{T}_1(\vec f^*_2)-\hat{T}_2(\vec f^*_2)||_\mathcal{F}\\
				&\leq \alpha||\vec f^*_1-\vec f^*_2||_\mathcal{F}+||g(C(\vec f_2^*)+\vec J_1)-g(C(\vec f_2^*)+\vec J_2)||_\mathcal{F}\\
				&\leq \alpha||\vec f^*_1-\vec f^*_2||_\mathcal{F} + \beta ||\vec J_1 -\vec J_2 ||_\mathcal{J} 
			\end{split}
		}
		
		As $\vec f^*_1=h(\vec J_1)$ and $\vec f^*_2=h(\vec J_2)$, we have $||h(\vec J_1)-h(\vec J_2)||_\mathcal{F}\leq \frac{\beta}{1-\alpha} ||\vec J_1-\vec J_2||_\mathcal{J}$, which completes the proof.
	\end{proof}

	\subsection{Proof of Proposition \ref{prop:mono}} \label{sec: proof_pro_3}
	\begin{proof}
		For $i \in \mathcal L$, define $g_i: \Re^L  \rightarrow \Re$ as the mapping from link cost  to $i$-th link flow, $g_i(\vec c)=[g(\vec c)]_i$. As $\vec f^*$ is the fixed point of $T$, we have $f^*_{1,i}= g_i(C_1(\vec f^*_1))$ and $f^*_{2,i}= g_i(C_2(\vec f^*_2))$.
		
		From condition 1 in Lemma \ref{lemma:unique}, the link cost-flow function $C_1 (\vec f)$ is monotone non-decreasing, i.e. $[C_1(\vec f^*_1) - C_1(\vec f^*_2)]^\transpose [\vec f^*_1 - \vec f^*_2]\geq 0$.
		
		We want to prove the result by contradiction. Assume that on $l$-th link, $f^*_{1,l} \geq f^*_{2,l}$. From the assumption about link cost function in this proposition, we have $c_{1,l}(\vec f^*_2) - c_{2,l}(\vec f^*_2) \geq 0$. Therefore, $[f^*_{1,l} - f^*_{2,l}] [c_{1,l}(\vec f^*_2) - c_{2,l}(\vec f^*_2)] \geq 0$ holds.
		
		Then,
		\beq{ \begin{split}
				&[g(C_1(\vec f^*_1)) - g(C_2(\vec f^*_2))]^\transpose [C_1(\vec f^*_1) - C_2(\vec f^*_2)] \\
				&= \sum_{i \in \mathcal{L}} [f^*_{1,i} - f^*_{2,i}] [c_{1,i}(\vec f^*_1) - c_{2,i}(\vec f^*_2)] \\
				&= \sum_{i \neq l, i \in \mathcal{L}} [f^*_{1,i} - f^*_{2,i}] [c_{1,i}(\vec f^*_1) - c_{2,i}(\vec f^*_2)] +[f^*_{1,l} - f^*_{2,l}] [c_{1,l}(\vec f^*_1) - c_{2,l}(\vec f^*_2)]\\
				&= \sum_{i \neq l, i \in \mathcal{L}} [f^*_{1,i} - f^*_{2,i}] [c_{1,i}(\vec f^*_1) - c_{1,i}(\vec f^*_2)] \\
				&\textit{\space} \textit{\space} \textit{\space}+[f^*_{1,l} - f^*_{2,l}] [c_{1,l}(\vec f^*_1) - c_{1,l}(\vec f^*_2) + c_{1,l}(\vec f^*_2) - c_{2,l}(\vec f^*_2)]\\
				&= \sum_{i \in \mathcal{L}} [f^*_{1,i} - f^*_{2,i}] [c_{1,i}(\vec f^*_1) - c_{1,i}(\vec f^*_2)] +[f^*_{1,l} - f^*_{2,l}] [c_{1,l}(\vec f^*_2) - c_{2,l}(\vec f^*_2)]\\
				&=[C_1(\vec f^*_1) - C_1(\vec f^*_2)]^\transpose [\vec f^*_1 - \vec f^*_2]+[f^*_{1,l} - f^*_{2,l}] [c_{1,l}(\vec f^*_2) - c_{2,l}(\vec f^*_2)] \\
				&\geq 0
			\end{split}
		}
		
		From the quasi strict monotonicity in Lemma 3 \cite{cantarella1997general}, we have $[g(C_1(\vec f^*_1)) - g(C_2(\vec f^*_2))]^\transpose [C_1(\vec f^*_1) - C_2(\vec f^*_2)] < 0$, which contradicts the above result and completes the proof, i.e., $f^*_{1,l} < f^*_{2,l}$
		
	\end{proof}

	\section{Proofs in Section 5} \label{sec: proof_sec_5}
	
	\subsection{Proof of Lemma \ref{lem:f} }
	\begin{proof} \label{sec: proof_lem_4}
		The convergence follows \cite[Theorem 3]{cantarella1997general} by replacing $\nabla \varphi(\vec f)$ with $C(\vec f)- C(\vec f^*)+\vec J_o$.

		If there exists $\vec f' \neq \vec f^*$ such that $\vec f'=h(\vec J_o)$, then both $\vec f'$ and $\vec f^*$ are the equilibrium flow under link cost function $C(\vec f)+ \vec J_o$, which contradicts the uniqueness from Lemma \ref{lemma:unique}.
	\end{proof}

	\subsection{Proof of Theorem \ref{thm:twotimescale}} \label{sec: proof_thm_1}
	\begin{proof}
		For the two time-scale algorithm, since $\beta_k=o(\alpha_k)$, the iteration \ref{equ:algo_j} proceeds at a slow rate compared to the iteration \ref{equ:algo_f}. According to Theorem 1.1 \cite{borkar1997stochastic}, $(\vec J_k, \vec f_k)$ converges to $(\vec J^*, \vec h(\vec J^*))$ on the set $\{\sup_k \vec J_k < \infty, \sup_k \vec f_k < \infty\}$ if the following two assumptions hold:
		\begin{enumerate}
			\item The o.d.e.
			\beq{
				\label{equ:ode_f}
				\dot{\vec f} = g(C(\vec f)+\vec J) -\vec f
			}
			has a unique global asymptotically stable equilibrium $\vec f^*=h(\vec J)$ such that $h$ is Lipschitz.
			\item The o.d.e.
			\beq{
				\label{equ:ode_j}
				\dot{\vec J} = \psi(h (\vec J), \vec J) -\vec J
			}
			has a unique global asymptotically stable equilibrium $\vec J^*$.
		\end{enumerate}
		
		According to the Lemma \ref{lem:f}, for any $\vec J_0$, there exists a unique $\vec f^*=h(\vec J_o)$ and $h$ is Lipschitz from Proposition \ref{prop:lip}. That is, we only need to show that the o.d.e. in \eqref{equ:ode_j} has a unique stable solution.

		Let $\lambda(\vec J)=\psi(h(\vec J),\vec J)$, then we have the optimal incentive $\vec J^*=\lambda(\vec J^*)$. Define a function $\varphi: \mathbb{J} \rightarrow \Re $ such that $\varphi(\vec J)>0$ and $\nabla \varphi(\vec J)=h(\vec J^*)-h(\vec J)$. Define $\dot{\varphi}= \langle \nabla \varphi(\vec J),\lambda(\vec J)-\vec J \rangle = \nabla \varphi(\vec J)^T[\lambda(\vec J)-\vec J], \forall \vec J \in \mathbb{J}$.
		
		As $\nabla \varphi(\vec J^*)=0$ and $\lambda(\vec J^*)=\vec J^*$,  $\nabla \varphi(\vec J^*)^T[\lambda(\vec J^*)-\vec J^*]=0$ holds. From Assumption 1, we have
		$\nabla \varphi(\vec J)^T[\lambda(\vec J)-\vec J]=(h(\vec J^*)-h(\vec J))^T [\psi(h(\vec J),\vec J)-\vec J]<0,$ for all $ \vec J \neq \vec J^*$.  Therefore, we can conclude that $\dot{\varphi} \leq 0, \forall \vec J \in \mathbb{J}$, with strict inequality when $\vec J \neq \vec J^*$.
		
		Let $\mathbb{A} \subset \mathbb{J}$ be the set of points of where $\dot{\varphi}$ vanishes, i.e., $\mathbb{A}= \{\vec J \in \mathbb{J}: \dot{\varphi}=0 \}$. Let $\mathbb{M} \subset \mathbb{A}$ be the largest set contained in $\mathbb{A}$ that is invariant under \eqref{equ:ode_j}, thus $\mathbb{M}=\{\vec J^*\}$.
		
		By LaSalle's invariance principle \cite[Theorem 3.4]{khalil2002nonlinear}, we have the conclusion that $\lim_{t \rightarrow \infty} dist(\vec J(t), M) = 0$, i.e., every solution starting in $\mathbb{J}$ approaches $M$ as $t \rightarrow \infty$. As $M$ is singleton and $\mathbb{J}= \Re^L$, we conclude that there exists a unique global asymptotically stable equilibrium $\vec J^*$.
	\end{proof}

	\section{Proofs in Section 6} \label{sec: proof_sec_6}
	
	\subsection{Proof of Theorem \ref{thm:sharing}} \label{sec: proof_thm_2}
	\begin{proof}
		From Section \ref{sec:pricing}, the total profit after cooperation $R_c$ is maximized, which leads to the inequality $\sum_{i \in \mathbb{S}} t_i < R_c$, i.e., the total platform profit after cooperation $R_c$ is larger than the profit before cooperation $\sum_{i \in \mathbb{S}} t_i$, where $t_i$ is the profit of provider $i$ before cooperation.
		
		As $\vec \theta = (\theta_1, \ldots, \theta_\mathcal{S})$ is the weight vector, there is a unique bargaining solution $\vec R^*$ that satisfies the above axioms (page 38) \cite{binmore1992fun}  (page 36)\cite{muthoo1999bargaining} (Asymmetric Nash Bargaining solution), where 
		\beqq{ \begin{split}
				\vec R^* &=\underset{\vec R \in F}{\arg \max} \sum_{i \in \mathbb{S}} \theta_i \log(R_i -t_i)\\
				\text{such that \space} &  R_i \geq t_i \forall i \in \mathbb{S}
			\end{split}
		}
		
		Since $\sum_{i \in \mathbb{S}} t_i < R_c$, the feasible region of above optimization problem is not empty.
		
		Let $y_i=R_i-t_i$, then it is equivalent to solve
		\beqq{ \begin{split}
				&\arg \max \sum_{i \in \mathbb{S}} \theta_i \log(y_i)\\
				\textit{s.t. \space} & \sum_{i \in \mathbb{S}} y_i = R_c -\sum_{i \in \mathbb{S}} t_i\\
				& y_i \geq 0, \forall i \in \mathbb{S}
			\end{split}
		}
		
		For the objective function, we have $\arg \max \sum_{i \in \mathbb{S}} \theta_i \log(y_i)=\arg \max \prod_{i \in \mathbb{S}} y_i^{\theta_i}$.
		From the equality constraint, we have $y_\mathcal{S}=R_c-\sum_{i}^\mathcal{S} t_i- \sum_{i=1}^{\mathcal{S}-1} y_i$,
		we define $\vec y=(y_1, \ldots, y_{\mathcal{S}-1})$ and function $\varphi(\vec y)=\prod_{i=1}^\mathcal{S} y_i^{\theta_i}$.
		
		From the first-order optimality $\nabla \varphi(\vec y^*)=0$ (the domain is convex, $\nabla^2 \varphi(\vec y)$ positive definite), we have 
		\beqq{
			\frac{y_1^*}{\theta_1}=\frac{y_\mathcal{S}^*}{\theta_\mathcal{S}}, \ldots, \frac{y_{\mathcal{S}-1}^*}{\theta_{\mathcal{S}-1}}=\frac{y_\mathcal{S}^*}{\theta_\mathcal{S}}
		}
		Therefore, we have 
		\beqq{
			y_i^*=\frac{\theta_i}{\sum_{j \in \mathbb{S}} \theta_j}\left(R_c-\sum_{j \in \mathbb{S}} t_j\right)
		}
		and $\vec R^*$ follows $R_i^*=y_i^*+t_i$.
	\end{proof}

\end{document}

%% file: hyperpath_fig.tex
\begin{tikzpicture}[scale=0.95]

\node[state] at (0, 0) (o) {$o$};

\node[state] at (11, 0) (d) {$d$};

\node[state] at (3, 0) (1) {$1$};
\node[state] at (3, -2) (2) {$2$};
\node[state] at (6, -0) (3) {$3$};

\node[state] at (9, -2) (4) {$4$};
\node[state] at (9, -4) (5) {$5$};

\draw[every loop, thick]
(o) edge[solid, green, bend left,auto = left] node {$(\$30, 27')$} (d)
(o) edge[densely dotted, cyan, auto = left] node {$( \$2, 20')$} (1)
(1) edge[densely dotted, cyan, auto = left] node {$( \$2, 22')$} (3)
(3) edge[densely dotted, cyan, bend left, auto = right] node {$( \$2, 17')$} (4)
(4) edge[densely dotted, cyan, bend right, auto = right] node {$( \$3, 15')$} (d)
(4) edge[densely dotted, cyan, bend left, auto = left] node {$( \$2, 25')$} (d)
(o) edge[dashdotted, blue, bend right, auto = right] node {$( \$1, 21')$} (2)
(5) edge[dashdotted, blue, bend right, auto = right] node {$( \$1, 22')$} (d)
(2) edge[dashed, red, bend right, auto = left] node {$( \$3, 25')$} (5)
(1) edge[double, black, bend right, auto = right] node {$( \$2, 5')$} (2)
(2) edge[double, black, bend right, auto = right] node {$( \$2, 5')$} (1)
(4) edge[double, black, bend right, auto = right] node {$( \$2, 5')$} (5)
(5) edge[double, black, bend right, auto = right] node {$( \$2, 5')$} (4);


\coordinate (L) at (-1.5,-6);
\draw[solid, green, thick] (L) -- + (1,0)               node[right] {: Taxi};
\draw[densely dotted, cyan, thick]  ([xshift=25mm] L) -- + (1,0) node[right] {: Bus};
\draw[dashed, red, thick]  ([xshift=50mm] L) -- + (1,0) node[right] {: Subway};
\draw[dashdotted, blue, thick]  ([xshift=80mm] L) -- + (1,0) node[right] {: Bike};
\draw[double, black, thick]  ([xshift=105mm] L) -- + (1,0) node[right] (LL) {: Scooter};
\node[draw, yshift=0.5ex, fit=(L) (LL)] {};
\end{tikzpicture}

%% file: hyperpath_fig3.tex
\begin{tikzpicture}[scale=0.95]

\node[state] at (0, 0) (o) {$o$};

\node[state] at (11, 0) (d) {$d$};

\node[state] at (3, 0) (1) {$1$};
\node[state] at (3, -2) (3) {$3$};

\node[state] at (6, -0.5) (2) {$2$};
\node[state] at (6, -2.5) (4) {$4$};

\draw[every loop, thick]
(o) edge[solid, green, bend left,auto = left] node {$1:(\$50, 44')$} (d)
(o) edge[solid, green, auto = left] node {$2:(\$20, 14')$} (1)
(o) edge[densely dotted, bend right, cyan, auto = right] node {$10:(\$3, 48')$} (3)
(1) edge[densely dotted, cyan, bend left, auto = left] node {$3:(\$3, 42')$} (2)
(2) edge[densely dotted, cyan, bend left, auto = left] node {$\textit{\space} 4:(\$3, 38')$} (d)
(2) edge[densely dotted, cyan, bend right, auto = left] node {5:$(\$4, 36')$} (d)
(3) edge[dashed, red, bend right, auto = right] node {$11:(\$5, 34')$} (4)
(4) edge[dashed, red, bend right, auto = right] node {$12:(\$4, 40')$} (d)
(1) edge[double, black, bend right, auto = right] node {$6:(\$1, 6')$} (3)
(3) edge[double, black, bend right, auto = right] node {$7:(\$1, 6')$} (1)
(2) edge[double, black, bend right, auto = right] node {$8:(\$1, 6')$} (4)
(4) edge[double, black, bend right, auto = right] node {$9:(\$1, 6')$} (2);


\coordinate (L) at (0,-4.5);
\draw[solid, green, thick] (L) -- + (1,0)               node[right] {: Taxi};
\draw[densely dotted, cyan, thick]  ([xshift=25mm] L) -- + (1,0) node[right] {: Bus};
\draw[dashed, red, thick]  ([xshift=50mm] L) -- + (1,0) node[right] {: Subway};
\draw[double, black, thick]  ([xshift=80mm] L) -- + (1,0) node[right] (LL) {: Scooter};
\node[draw, yshift=0.5ex, fit=(L) (LL)] {};
\end{tikzpicture}